%% file: Arxiv.tex
\documentclass[runningheads]{llncs}
\usepackage[utf8]{inputenc}

\usepackage{amsmath,amssymb,amsfonts,}
\usepackage{graphicx}
\usepackage{latexsym}
\usepackage{algorithm}
\usepackage[noend]{algorithmic}
\usepackage{multirow,eepic}
\usepackage{wrapfig}
\usepackage{pdfpages}
\usepackage{refcount}

\makeatletter
\def\senbun#1(#2)#3({\@senbun(#2)(}
\def\@senbun(#1,#2)(#3,#4){%
   \@tempdima#1\p@ \advance\@tempdima#3\p@
   \divide\@tempdima\tw@
   \@tempdimb#2\p@ \advance\@tempdimb#4\p@
   \divide\@tempdimb\tw@
   \edef\@senbun@temp{\noexpand\qbezier(#1,#2)%
      (\strip@pt\@tempdima,\strip@pt\@tempdimb)(#3,#4)}%
   \@senbun@temp}
\makeatother
\newcommand{\ASY}{{\sc Asynch}}
\newcommand{\FSY}{{\sc Fsynch}}
\newcommand{\SSY}{{\sc Ssynch}}
\newcommand{\RSY}{{\sc Rsynch}}

\newcommand{\LU}{{\mathcal{LUMI}}} 
\newcommand{\LUM}{{\mathcal{LUMI}}} 
\newcommand{\FS}{{\mathcal{FSTA}}} 
\newcommand{\FC}{{\mathcal{FCOM}}} 
\newcommand{\OB}{{\mathcal{OBLOT}}} 
\newcommand{\N}{{\rm I\kern-.22em N}} 
\newcommand{\Z}{{\sf Z\kern-.42em Z}} 
\newcommand{\R}{{\rm I\kern-.22em R}}

\newcommand{\LK}{{\mathit{Look}}}
\newcommand{\CP}{{\mathit{Comp}}}
\newcommand{\M}{{\mathit{Move}}}

\newcommand{\T}{t}

\newcounter{Codeline}

\newtheorem{theoremduplicate}{Theorem}

\begin{document}
%
\title{Efficient Self-stabilizing Simulations of Energy-Restricted 
Mobile Robots by Asynchronous Luminous Mobile Robots\thanks{This work was supported in part by JSPS KAKENHI Grant Number~20K11685, and~21K11748.}}
\titlerunning{Efficient Self-stabilizing Simulations}
%
\author{Keita Nakajima\inst{1}\orcidID{0009-0003-2535-653X} 
 \and
Kaito Takase\inst{2}\orcidID{0009-0002-9757-5371}
\\ \and
Koichi Wada\inst{2}\orcidID{0000-0002-5351-1459}
}
\authorrunning{K. Nakajima et al.}
%
\institute{Tokyo Institute of Technology, Tokyo, Japan\\
\email{nakajima.k.au@m.titech.ac.jp} \and
Hosei University, Tokyo, Japan \\
\email{
kaito.takase.6z@stu.hosei.ac.jp, wada@hosei.ac.jp}
}
\maketitle              
\begin{abstract}
In this study, we explore efficient simulation implementations to demonstrate computational equivalence across various models of autonomous mobile robot swarms. Our focus is on \RSY, a scheduler designed for energy-restricted robots, which falls between \FSY\ and \SSY. We propose efficient protocols for simulating $n(\geq 2)$ luminous ($\LU$) robots operating in \RSY\ using $\LU$ robots in \SSY\ or \ASY. Our contributions are twofold:
\begin{enumerate}
\item[(1)] We introduce protocols that simulate $\LU$ robots in \RSY\ using $4k$ colors in \SSY\ and $5k$ colors in \ASY, for algorithms that employ $k$ colors. This approach notably reduces the number of colors needed for \SSY\ simulations of \RSY, compared to previous efforts. Meanwhile, the color requirement for \ASY\ simulations remains consistent with previous \ASY\ simulations of \SSY, facilitating the simulation of \RSY\ in \ASY.
\item[(2)] We establish that for $n=2$, \RSY\ can be optimally simulated in \ASY\ using a minimal number of colors.
\end{enumerate}
Additionally, we confirm that all our proposed simulation protocols are self-stabilizing, ensuring functionality from any initial configuration.
\end{abstract}
\newpage
\section{Introduction}

\input{introduction}

\section{Preliminaries}
\label{sec:preliminaries}
\subsection{Robots}

\input{Prelim-Basics}

\input{Prelim-Models}

\input{Prelim-Schedulers}

\subsection{Computational Relationships}

\input{Prem-ComRel}

\section{Simulation of \RSY\ on $\LU$}\label{sec:sim-RSY-on-LUMI}
\subsection{4-Color Simulation of \RSY\ by \SSY}
\input{arxiv/4Col-RS-SS}

\subsection{Making $\textup{SIM}^{RS}_{S}$ Self-stabilizing}
\input{arxiv/4Col-self-stable}

\subsection{5-Color Simulation of \RSY\ by \ASY}
\input{arxiv/5Col-RS-AS}

\subsection{Making $\textup{SIM}^{RS}_{A}$ Self-stabilizing}
\input{arxiv/5Col-self-stable}

\section{Optimal Simulation of \RSY\ by \ASY\ with Two Robots}\label{sec:sim-2-RSY-by-ASY}
\input{arxiv/3col-RS-AS-2robots}



\section{Concluding Remarks}\label{sec:conclusion}
\input{conclusion}

\newpage
\bibliographystyle{plainurl}
\bibliography{referenceorg}

\clearpage







\end{document}

%% file: introduction.tex
\subsection{Background and Motivation}
The computational issues of autonomous mobile entities operating  in a Euclidean space in $\mathit{Look}$-$\mathit{Compute}$-$\mathit{Move}$ ($\mathit{LCM}$) cycles
have been the subject of extensive research in  distributed computing. 
 In the $\mathit{Look}$ phase, an entity, viewed as a point and usually called {\em robot},  obtains a snapshot of the space; in  the $\mathit{Compute}$ phase
it  executes its algorithm (the same for all robots) using the snapshot as input; it then moves towards the computed destination in the $\mathit{Move}$ phase.
Repeating these cycles, the robots can collectively perform some tasks and solve some problems. 
The research interest has been on determining  the impact that
{\em internal}  capabilities (e.g., memory, communication) and {\em external}
 conditions (e.g. synchrony, activation scheduler) have on the solvability of a problem.

 In the most common model, $\OB$, in addition to the standard assumptions of {\em anonymity} and {\em uniformity} (robots have no IDs and run identical algorithms),  the  robots are 
 {\em oblivious} (no persistent memory to record information of previous cycles)
 and  {\em silent} (without explicit means of communication). 
 Computability in this model has been the object of intensive research since its
  introduction in  \cite{SY}. 
 Extensive investigations have been carried out to clarify the computational limitations and powers of these robots for basic coordination tasks such as Gathering (e.g., \cite{AP,AOSY,BDT,CDN,CFPS,CP,FPSW05,ISKIDWY,SY}), Pattern Formation (e.g., \cite{FPSW08,FYOKY,SY,YS,YUKY}), Flocking (e.g., \cite{CG,GP,SIW}).

A model which provides robots with  persistent memory, albeit limited,  and
communication means is the $\LU$ model, formally  defined and analyzed  in \cite{DFPSY}, following a suggestion in \cite{PD}.
In this model, each robot is equipped with a constant-sized memory (called {\em light}),
 whose value (called {\em color}) can be set during the {\em Compute} phase. The light  is  visible to all the robots and is persistent in the sense that it is not automatically reset at the end of a cycle. Hence,  these luminous robots  are capable  of both remembering and communicating a constant number of bits in each cycle. 

 An important result is that,  despite these limitations, the simultaneous presence  of
 both persistent memory and communication renders
luminous robots strictly more powerful than oblivious robots \cite{DFPSY}.
This, in turns, has opened the question about the individual computational power of the two internal capabilities, memory and communication,
and motivated the investigations on 
two sub-models of $\LU$: $\FS$ where 
the robots have a constant-size persistent memory but are silent, 
and $\FC$,  where robots can communicate a constant number 
of bits but are oblivious (e.g., see \cite{BFKPSW22,apdcm,FSVY,FSW19,OWD,OWK}).  

 
\ \\
All these studies across various models have highlighted the crucial role played by two interrelated {\em external} factors:
 the level of synchronization and the activation schedule provided by the system.
 As in other types of distributed computing systems, there are two different settings; the 
synchronous and the asynchronous settings.
In the {\em synchronous} (also called  {\em semi-synchronous})   (\SSY) setting, introduced in \cite{SY}, time   is divided into discrete 
intervals, called {\em rounds}. In each round, an arbitrary but nonempty subset 
of the robots is activated, and they simultaneously perform
exactly one atomic $\LK$-$\CP$-$\M$ cycle. The selection of which robots are 
activated at a given round is made by an adversarial scheduler, constrained only 
to be fair, i.e., every robot is activated infinitely often. 
Weaker form of synchronous adversaries have also been introduced and investigated. The most important and 
extensively studied is the {\em fully-synchronous} (\FSY) scheduler, which activates all the robots  in every round.
Other  interesting synchronous schedulers are
 \RSY, studied for its use to model energy-restricted robots \cite{BFKPSW22},  as well as
 the family of {\em Sequential} schedulers  (e.g., {\em Round Robin}), 
where in each round only one robot is activated. 

In the {\em asynchronous} setting (\ASY), introduced in  \cite{FPSW99}, 
 there is no common notion of time and each robot is activated 
 independently of the others. it allows for  arbitrary but finite delays between the $\LK$, $\CP$ and $\M$ phases, and each movement may take an arbitrary but finite amount of time. The duration of each robot's cycle, as well as the timing of robot's activation,  are
controlled by an adversarial scheduler, constrained only to be fair,
i.e., every robot must be activated infinitely often.

%
%
%

\subsection{Contributions}
Like in other types of distributed systems, understanding the computational difference between various levels of synchrony and asynchrony has been a primary research focus.
In the robot model, to "separate" between the computational power of robots in two settings, we demonstrate that certain problems are solvable in one model but unsolvable in another. For example, 
 in $\OB$, {\em Rendezvous} is unsolvable in \SSY, but solvable in \FSY \cite{SY}, indicating a separation between \FSY\ and \SSY\ in $\OB$. Conversely,  
to show that a weaker model is equivalent to a stronger model, we devise a {\em simulation protocol} that allows the correct execution of any protocol from the stronger model in the weaker model. The first attempt in the robot model
 involves constructing a simulation protocol for any $\LU$ protocol in \SSY\ using $\LU$ robots in \ASY~\cite{DFPSY}. 
 This protocol employs $5k$ colored lights in $\LU$ robots to simulate \SSY\ protocols using $k$ colors\footnote{In the case of $k=1$, this protocol simulates $\OB$ robots working in \SSY\ with $\LU$ robots with $5$ colors in \ASY.}. 
 
 In this paper, we focus on making such simulations more efficient, specifically, considering simulations that involve $\LU$ 
 robots operating in \RSY\ and $\LU$
robots operating in \ASY\, under the most unrestricted adversary.
Though \RSY\ is introduced for modeling energy-restricted robots~\cite{BFKPSW22}, it is interesting in its own right because $\LU$ robots in \ASY\ have the same power as those in \RSY~\cite{DFPSY,BFKPSW22}, and $\FC$ robots in \RSY\ have the same power as $\LU$ robots in \RSY~\cite{BFKPSW22}.

The simulator of $\LU$ robots with $k$ colors in \SSY\ uses $\LU$ robots in \ASY\ and utilizes $5k$ colors~\cite{DFPSY}, including $5$ colors to control the simulation. Therefore, we aim to reduce the number of colors used to control the simulation.

Table~\ref{tab:Prev-results} presents both the previous simulation results and our new results. Here for any model, $M \in\{\FC, \LU \}$ and 
any adversarial scheduler  $A\in\{$ \RSY, \SSY, \ASY $\}$,  $M^A$ denotes the robot model $M$ working in $A$. 

The first simulation protocol was designed for $\LU$ robots in \SSY\ using $\LU$ robots under the strongest adversary, \ASY. This simulation ensures that in every round, the only selection of which robots are 
activated is made by an adversarial scheduler, constrained only 
to be fair~\cite{DFPSY}. The simulation employs $5$ light colors to control the process.\ 
A unique property of this simulator is that not only does the simulated protocol function in \SSY, but it also ensures that any robot is activated exactly once during a certain duration, maintaining fairness in the scheduler. We will leverage this property to decrease the number of colors used when simulating $\LU$ robots in \RSY. When simulating a protocol involving $k$-color $\LU$ robots in an unfair \SSY, the simulator will use $3k$ colors for $\LU$ robots in \ASY~\cite{NSW-2021}.


In this paper, we first establish that
by utilizing the property of the simulator detailed in~\cite{DFPSY}, 
\begin{enumerate}
        \item[(1)] $k$-color $\LU$ robots in \RSY\ can be simulated by $4k$-color  $\LU$ robots in \SSY.
        \item[(2)] $k$-color $\LU$ robots in \RSY\ can be simulated by $5k$-color $\LU$ robots in \ASY.
\end{enumerate}
Previously, case (1) required $36k$ colors~\cite{BFKPSW22}. In contrast, our simulator for case (2) also uses $5k$ colors, effectively simulating $\LU$ robots in \RSY\ using $\LU$ robots in \ASY.
Additionally,
we demonstrate that when the number of robots is limited to $2$ ($n=2$), the simulator in the case (1) can be implemented more efficiently, Specifically, we show that:
\begin{enumerate}
    \item[(3)] $k$-color $\LU$ robots in \RSY\ can be simulated by $3k$-color $\LU$ robots in \ASY. 
    \item[(4)] In the case of $k=1$, $\OB$ robots in \RSY\ cannot be simulated by $2$-color $\LU$ robots in \ASY. This demonstrates that the number of colors used in the simulator shown in (3) is optimal.
\end{enumerate}
We also confirm that all our proposed simulation protocols are self-stabilizing, ensuring functionality from any initial configuration. These self-stabilization can be done without increasing the number of colors.

The remainder of the paper is organized as follows.
In Section~\ref{sec:preliminaries}, we present robot models, schedulers, and the preliminaries used in this paper. In Sections~\ref{sec:sim-RSY-on-LUMI},  we show simulation protocols of \RSY\ by \SSY\ and \ASY\ on $\LU$, respectively.  Section~\ref{sec:sim-2-RSY-by-ASY} provides optimal simulation of \RSY\ by \ASY\ with two robots. Finally, in Section~\ref{sec:conclusion}, we conclude the paper.


{\small
\begin{table}[ht]
  \caption{The previous results and this paper's results.}
\hspace{2cm} $n \geq 2$
   \begin{center}  
   
        \begin{tabular}{|c|c|c|c|}
\hline
simulating model        & simulated model  & $\#$ colors  & Ref.   
      \\ \hline
 $\mathcal{\LU}^{A}$   & $\mathcal{\LU}^{S}$ 
    & $5k$ & \cite{DFPSY} \\ \hline
  $\mathcal{\LUM}^{S}$  & $\mathcal{\LUM}^{RS}$ & $36k$ 
   &  \cite{BFKPSW22}\\ \hline
   $\mathcal{\LUM}^{A}$  & $\mathcal{\LUM}^{S}$ $^*$ & $3k$ 
   &  \cite{NSW-2021}\\ \hline  
   $\mathcal{\FC}^{F}$ & $\mathcal{\LUM}^{F}$ & $2k^2$& \cite{BFKPSW22}  \\ \hline
      $\mathcal{\FC}^{RS}$ & $\mathcal{\LUM}^{RS}$ & $64k2^k$&  \cite{BFKPSW22} \\ \hline
      $\mathcal{\LUM}^{S}$ & $\mathcal{\LUM}^{RS}$ & $4k$& This paper  \\ \hline
            $\mathcal{\LUM}^{A}$ & $\mathcal{\LUM}^{RS}$ & $5k$& This paper  \\ \hline

\end{tabular}
    \end{center}
\hspace{2cm} $*$ unfair \SSY
\vspace{3mm}
       
        \hspace{2cm} $n = 2$
         \begin{center}
        \begin{tabular}{|c|c|c|c|}
\hline
simulating model        & simulated model  & $\#$ colors  & Ref.   
      \\ \hline
 $\mathcal{\LU}^{A}$   & $\mathcal{\LU}^{RS}$ 
    & $3k$ & This paper \\ \hline

\end{tabular}
    \end{center}

  \label{tab:Prev-results}
 %
 %

\end{table}
}

%% file: Prelim-Basics.tex
The systems considered in this paper consist of a team  $R = \{ r_0 ,\cdots,
r_{n-1}\}$ of  computational entities moving and operating
 in the Euclidean plane $\mathbb R^2$. Viewed as points and called {\em robots},
the entities can move freely and continuously in the plane.
Each robot has its own local coordinate system and it always perceives itself at its origin;
there might not be consistency between the coordinate systems of the robots.
A robot is equipped with sensorial devices that allow it to observe the positions of the other robots in its local coordinate system.

Robots are {\em identical}: they are indistinguishable by their appearance, and they execute the same protocol. Robots are {\em autonomous}, without central control.  

At any time, a robot is {\em active} or {\em inactive}. Upon becoming active, a robot $r$ executes a $\mathit{ Look}$-$\mathit{Compute}$-$\mathit{Move}$ ($\mathit{LCM}$) cycle performing the following three operations:
\begin{enumerate}
\item {\em Look:} The robot activates its sensors to obtain a snapshot of the positions occupied by the robots with respect to its own coordinate
system\footnote{This is called the {\em full visibility} (or unlimited visibility)  setting; restricted forms of visibility have also been considered for these systems~\cite{FPSW05}}.
\item {\em Compute:} The robot executes its algorithm using the snapshot as input. The result of the computation is a destination point.
\item {\em Move:} The robot moves in a straight line towards the computed destination;
if the destination is the current location, the robot stays still.
 \end{enumerate}
 
\noindent When inactive, a robot is idle. All robots are initially idle. The time it takes to complete a cycle is assumed to be finite and the operations $\mathit{Look}$
 and $\mathit{Compute}$  are  assumed to be instantaneous.

%% file: Prelim-Models.tex
In the standard model, $\OB$, the robots are {\em silent}: 
they have no explicit means of communication; 
furthermore, they are {\em oblivious}: at the start of a cycle, a robot has no
memory of observations and computations performed in previous cycles.

In the other common model, $\LUM$, 
each robot $r$ is equipped with a persistent variable of visible 
state $Light[r]$, called {\em light}, whose values are taken from a finite 
set $C$ of states called {\em colors} (including the color that 
represents the initial state when the light is off). 
The colors of the lights can be set in each cycle by $r$ at 
its {\em Compute} operation. 
A light is {\em persistent} from one computational cycle to the next: 
the color is not automatically reset at the end of a cycle;  
the robot is otherwise oblivious, forgetting all other information from previous cycles. 
If any color is not set to some light, the color of the light remains unchanged.
In $\LUM$, the {\em Look} operation produces a colored snapshot; 
i.e., it returns the set of pairs 
 $(position,color)$ 
of the other robots\footnote{If  (strong) multiplicity detection is assumed, 
the snapshot is a multi-set.}.
Note that if $|C|=1$, then the light is not used; thus, this case corresponds to the $\OB$ model. 

 
%
%
%
 In all the above models,
 a {\em configuration} ${\cal C}(\T)$ at time $\T$ is the multiset of the $n$ pairs 
 $(x_i(\T), c_i(\T))$, where $c_i(\T)$ is the color of robot $r_i$ at time $\T$.

%% file: Prelim-Schedulers.tex
\subsection{Schedulers, Events}

With respect to the activation schedule of the robots, and the duration of their $\mathit{LCM}$  cycles, the fundamental distinction is between the {\em synchronous} and {\em asynchronous} settings.

In the {\em synchronous} setting (\SSY), also called {\em semi-synchronous} and first studied in \cite{SY},  time is divided into discrete
intervals, called {\em rounds}; in each round, a non-empty set of robots is activated  and  they simultaneously perform
a single  $\LK$-$\CP$-$\M$ cycle in perfect synchronization. 
The selection of which robots are 
activated at a given round is made by an adversarial scheduler, constrained only 
to be fair
(i.e., every robot is activated infinitely often). 
The particular  synchronous setting, where every robot is  activated in every round
is called {\em fully-synchronous} (\FSY).
In a synchronous setting, without loss of generality,  the expressions ``$i$-th round" and ``time $t=i$" are used as synonyms.

In the {\em asynchronous} setting (\ASY), first studied in \cite{FPSW99}, 
 there is no common notion of time,  the duration of each phase is finite 
 but unpredictable and might be different in different cycles, 
 and each robot is activated 
 independently of the others.
The duration of the phases of each cycle as well as the decision of when 
 a robot is activated is
controlled by an adversarial scheduler, constrained only to be fair,
i.e., every robot must be activated infinitely often.

In the asynchronous settings, the execution  by a robot of any of the operations  $\mathit{Look}$, $\mathit{Compute}$ and $\mathit{Move}$  is called an {\em event}.
We associate relevant time information to events:
for the  $\mathit{Look}$  (resp., $\mathit{Compute}$) operation,
which is instantaneous, the relevant time is $\T_L$  (resp., $\T_C$) when the event occurs;
 for the    $\mathit{Move}$ operation,   these are 
the times $\T_B$  and $\T_E$ when the event begins and ends, respectively.
Let ${\cal T} =\{\T_1, \T_2, ...\}$ denote the infinite ordered set of all relevant times; i.e.,
$\T_i < \T_{i+1}, i\in \N$.  In the following, to simplify the presentation and
without any loss of generality, we will refer to $\T_i\in {\cal T}$  simply
by its index $i$; i.e.,  the expression ``time $t$'' will be used to mean
``time $\T_t$''.

\color{black}

  Consider now  the synchronous  scheduler, we shall call \RSY,  obtained from \SSY\ by
adding the following {\em restricted-repetition
  condition} to its activation sequences:
  {\small
\begin{multline*}
\bigg[\forall i\geq 1, e_i=R\bigg] \textbf{ or }\\ \bigg[\exists p\geq 0 : \bigg(\big[\forall  i\leq p,  (e_i=R)\big] \textbf{ and } \big[\forall  i > p,( \emptyset\neq e_i\neq R \textbf{ and }  e_i \cap e_{i+1} = \emptyset)\big]\bigg)\bigg],
\end{multline*}
}%
where an {\em activation  sequence}  of $R$ 
  is  an infinite sequence  $E= \langle e_1, e_2, \ldots, e_i, \ldots \rangle$, and
$e_i\subseteq R$ denotes the set of robots
activated in round $i$.
That is,  this scheduler is composed of sequences where
 the prefix is a  (possibly empty) sequence of $R$ and, if the prefix is finite, the rest are non-empty sets
  satisfying the constraint $(e_i \cap e_{i+1} = \emptyset)$.

%% file: Prem-ComRel.tex
 Let ${\cal M} = \{\LU, \OB \}$  be the set of  models under investigation and 
${\cal S}= \{$ \RSY, \SSY, \ASY $\}$ be the set of  schedulers 
under consideration.

We denote by $\mathcal{R}$ the set of all robot teams that satisfy the core assumptions
(i.e. they are identical, autonomous and operate in $\mathit{LCM}$ cycles), and 
operate under rigidity of movements,  chirality, and  variable disorientation. 
By $\mathcal{R}_n \subset \mathcal{R}$ we denote the set of all teams of size $n$.

Given a model $M \in {\cal M}$, a scheduler $S\in  {\cal S}$,
and a team of robots $R\in\cal{R}$
we denote by $Task(M,k,S;R)$ 
 the set of problems solvable 
by $\mathit{R}$ in $M$ with $k$ colors 
under adversarial scheduler $S$.

%
%
%
%
%

 For simplicity of notation, 
 let  $M_k^{RS}(R), M_k^{S}(R)$, and $M_k^A(R)$ 
 denote\\ $Task(M,k,$\RSY$;\mathit{R})$,
 $Task(M,k,$\SSY$;\mathit{R})$, and $Task(M,k,$\ASY$;\mathit{R})$, \\ 
 respectively\footnote{Since $\OB$ robots have no light (one color), the suffix $k=1$ is omitted.}.

%


%% file: arxiv/4Col-RS-SS.tex



In this section, we show that semi-synchronous systems equipped with a light with $4$ colors are at least as powerful as \textit{restricted-repetition} system (\RSY) without lights. More precisely, we have: 

\begin{theorem}\label{4col-RS-SS}
$\forall R \in {\cal{R}}, \OB^{RS}(R) \subseteq \LU_4^{S}(R).$
\end{theorem}

We show that
every problem solvable by a set of $\OB$ robots  under \RSY\ can be also solved 
 by a set of $\LU$ robots with $4$ colors under \SSY.  We do so constructively: we present a {\em simulation} protocol 
for $\OB$ robots   that allows them to correctly execute in \RSY\ 
any protocol ${\cal P}$ given in input.

We first simulate  a restricted semi-synchronous scheduler called \textit{multiple-slicing}, in which any robot is activated exactly once in some duration, where the duration is called \textit{mega-cycle}. Then we modify the simulator working to attain the condition of \RSY.

A scheduler that a group of $n$ robots, starting from time $t$ = 0, after $n$ successive activation rounds (slices), all robots in the system will have been activated exactly once. This is called centralized slicing \SSY. We extend the centralized slicing \SSY\ to  $R_1^1, R_2^1, \ldots ,R_{k_1}^1;R_1^2, R_2^2, \ldots, R_{k_2}^2; \ldots ;R_1^i, R_2^i, \ldots, R_{k_i}^i; \ldots (1 \leq k_i \leq n)$, where for each $i \geq 1$, $R^i_1, R^i_2,\ldots, R^i_{k_i}$ are a partition of $R$. This scheduler is called {\em multiple-slicing}\footnote{It is centralized slicing if $|R^i_j|=1$ for every $i(\geq 1)$ and $j(1 \leq j \leq k_i$).} \SSY. At this time, $R^i_j$ is called the $j$-th {\em stage} in the $i$-th {\em mega-cycle}.

If the multiple slicing \SSY\ satisfies that $R^i_{k_i} \cap R^{i+1}_1 = \emptyset$ for every $i \geq 1$, this scheduler works to satisfy the {\em disjoint} condition of \RSY. 

Specifically, the robot should have one of the following colors:

\begin{itemize}
    \item \textbf{T(rying)}: denotes not having executed $\cal{P}$ in a current mega-cycle yet.
    \item \textbf{M(oving)}: denotes having already executed $\cal{P}$ once.
    \item \textbf{S(topped)}: denotes after executing $\cal{P}$ except in the last stage of a mega-cycle.
    \item \textbf{S'(topped)}: denotes after executing $\cal{P}$ in the last stage of a mega-cycle.
\end{itemize}

When a robot with $\alpha$ state, it is called an $\alpha$-robot. We also denote state set as a global configuration (e.g. $\forall T,S$ means each robot's state  either T or S, and there is at least one robot with such state).

Considering the states in this way, the rule of protocol can be considered as follows.

\paragraph*{Protocol Description}

Fig.~\ref{fig:alg4-rs-ss-tran} (a) shows the transition diagram representation of $\textup{SIM}^{RS}_{S}$. 
The protocol $\textup{SIM}^{RS}_{S}$ uses four colors: T, M, S, S'. Initially, all lights are set to T. 

The protocol simulates a sequence of mega-cycles, each of which starts with some robots trying to execute protocol ${\cal{P}}$ 
and ends with all robots finishing the mega-cycle having executed ${\cal{P}}$ once. 
After this end configuration, it transits to start one, and a new mega-cycle starts.

\begin{figure}[h]
\begin{center}
\includegraphics[width=120mm]{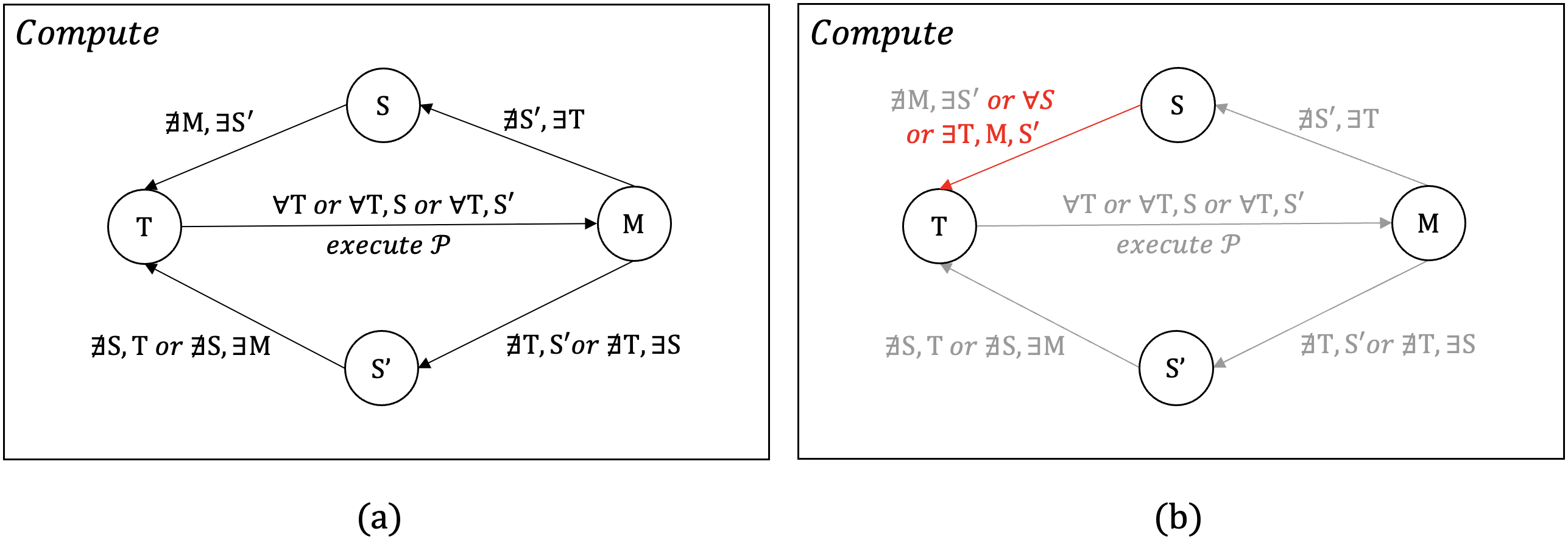}
\caption{Transition diagram representations of protocol $\textup{SIM}^{RS}_{S}$((a)), and self-stabilizing protocol ss-$\textup{SIM}^{RS}_{S}$((b)). The label in the nodes represents the color of the light of the executing robot. The label of an edge expresses the condition on the lights of all the other robots that must be satisfied for the transition to occur. The notation "$\forall A,B$" means: "$\{\textup{Light\lbrack{r}\rbrack}\ |\ \forall r \in R\} = \{A,B\}$", "$\exists A$" means: "$\exists r \in R, \textup{Light\lbrack{r}\rbrack} \in \{\textup{A}\}$", "$\nexists A$" means: "$\{\textup{Light\lbrack{r}\rbrack}\ |\ \forall r \in R\} \cap \{A\} = \emptyset$". Conditions, colored red in (b) are newly added.}\label{fig:alg4-rs-ss-tran}
\end{center}
\end{figure}

\begin{algorithm}[H]
    \caption{$\textup{SIM}^{RS}_{S}$}
    State \textit{Look}
    \begin{algorithmic}[0]
        \STATE $Pos\lbrack{r}\rbrack:$ the position on the plane of robot $r$ (according to my coordinate system);
        \STATE $Light\lbrack{r}\rbrack:$ the color of the light of robot.
    \end{algorithmic}
    (Note: I am robot $x$) \\ \\
    State \textit{Compute}
    \begin{algorithmic}[1]
        \STATE $p \leftarrow Pos\lbrack{x}\rbrack$.
        \STATE $c \leftarrow \{Light\lbrack{r}\rbrack\ |\ \forall r \in R\}$.
        \IF {$Light\lbrack{x}\rbrack$ = T}
            \IF {$c = \{\textup{T}\} \vee c = \{\textup{T,S}\} \vee c = \{\textup{T,S'}\}$}
                \STATE Execute ${\cal{P}}.$
                \STATE $p \leftarrow computed\ destination$.
                \STATE $Light\lbrack{x}\rbrack \leftarrow$ M.
            \ENDIF
        \ELSIF {$Light\lbrack{x}\rbrack$ = M}
            \IF {$c\cap\{\textup{S'}\} = \emptyset \wedge \{\textup{T}\} \subseteq c$}
                \STATE $Light\lbrack{x}\rbrack \leftarrow$ S.
            \ENDIF
            \IF {$c\cap\{\textup{T,S'}\} = \emptyset \vee (c\cap\{\textup{T}\} = \emptyset \wedge \{\textup{S}\} \subseteq c)$}
                \STATE $Light\lbrack{x}\rbrack \leftarrow$ S'.
            \ENDIF
        \ELSIF {$Light\lbrack{x}\rbrack$ = S} 
            \IF {$c\cap\{\textup{M}\} = \emptyset \wedge \{\textup{S'}\} \subseteq c$}
                \STATE $Light\lbrack{x}\rbrack \leftarrow$ T.
            \ENDIF
        \ELSIF {$Light\lbrack{x}\rbrack$ = S'}
            \IF {$c\cap\{\textup{S,T}\} = \emptyset \vee (c\cap\{\textup{S}\} = \emptyset \wedge \{\textup{M}\} \subseteq c)$}
                \STATE $Light\lbrack{x}\rbrack \leftarrow$ T.
            \ENDIF
        \ENDIF
    \vspace{3mm}
    \end{algorithmic} 
    
    State \textit{Move} 
    \begin{algorithmic}[0]
        \STATE Move($p$).
    \end{algorithmic}
    \label{alg4-rs-ss}
\end{algorithm}


During each mega-cycle, each robot gets the opportunity to execute one step of the protocol ${\cal{P}}$. A $T$-robot $r$, tries to execute protocol ${\cal{P}}$. However, the robot is allowed to execute ${\cal{P}}$ only if 
there are no $M$-robots (i.e. robots that executed  ${\cal{P}}$ before this round). If that is the case, $r$ changes its color to M. 
On the other hand, 
if there exist an $M$-robot, 
it does nothing (i.e. it waits until no $M$-robots exist). 
$M$-robots, after executing ${\cal{P}}$, will turn their own lights to $S$ only when 
no $S'$-robot exists and $T$-robots exist (which is in a stage except the last one), or turn to $S'$ only when 
no $c(\in\{T, S'\})$-robot exists (where all robots execute $\cal{P}$), or no $T$-robot exists and $S$-robots exist (which is in the last stage). If the robots turn to $S$, after some time, each robot will be colored either S (i.e. executed ${\cal{P}}$) or T (i.e. not executed ${\cal{P}}$), else all robot will be colored S' (i.e. this happens when all robots execute ${\cal{P}}$ at the same stage). 
\color{black}
At this time, $T$-robots are given another opportunity to execute ${\cal{P}}$. Thus, a cycle of protocol $\textup{SIM}^{RS}_{S}$ consists of several stages such that, in each stage, at least one robot executes ${\cal{P}}$ while other robots wait. Eventually all robots will colored either S or S' (i.e. each robot has executed ${\cal{P}}$ once), and the cycle ends when $S$-robots change to color T. At this point, the $S'$-robots do nothing; when this process is completed, a new cycle starts. When a new mega-cycle begins, only $T$-robots and $S'$-robots exist, and some of the $T$-robots perform ${\cal{P}}$ in the first stage. Here, since the $S'$-robots have executed ${\cal{P}}$ in the last stage of the previous mega-cycle, the robots performing ${\cal{P}}$ in the first stage and the $S'$-robots are mutually disjoint and the condition of \RSY\ is satisfied.
$S'$-robots will turn their own lights to T only when 
no $c(\in\{T, S\})$-robot exists, or no $S$-robot exists and $M$-robots exist, i.e. just before the first stage of the next mega-cycle. 

\begin{figure}[h]
\label{fig:alg-self-stable-rs-ss-state-tran}
\begin{center}
\includegraphics[width=110mm]{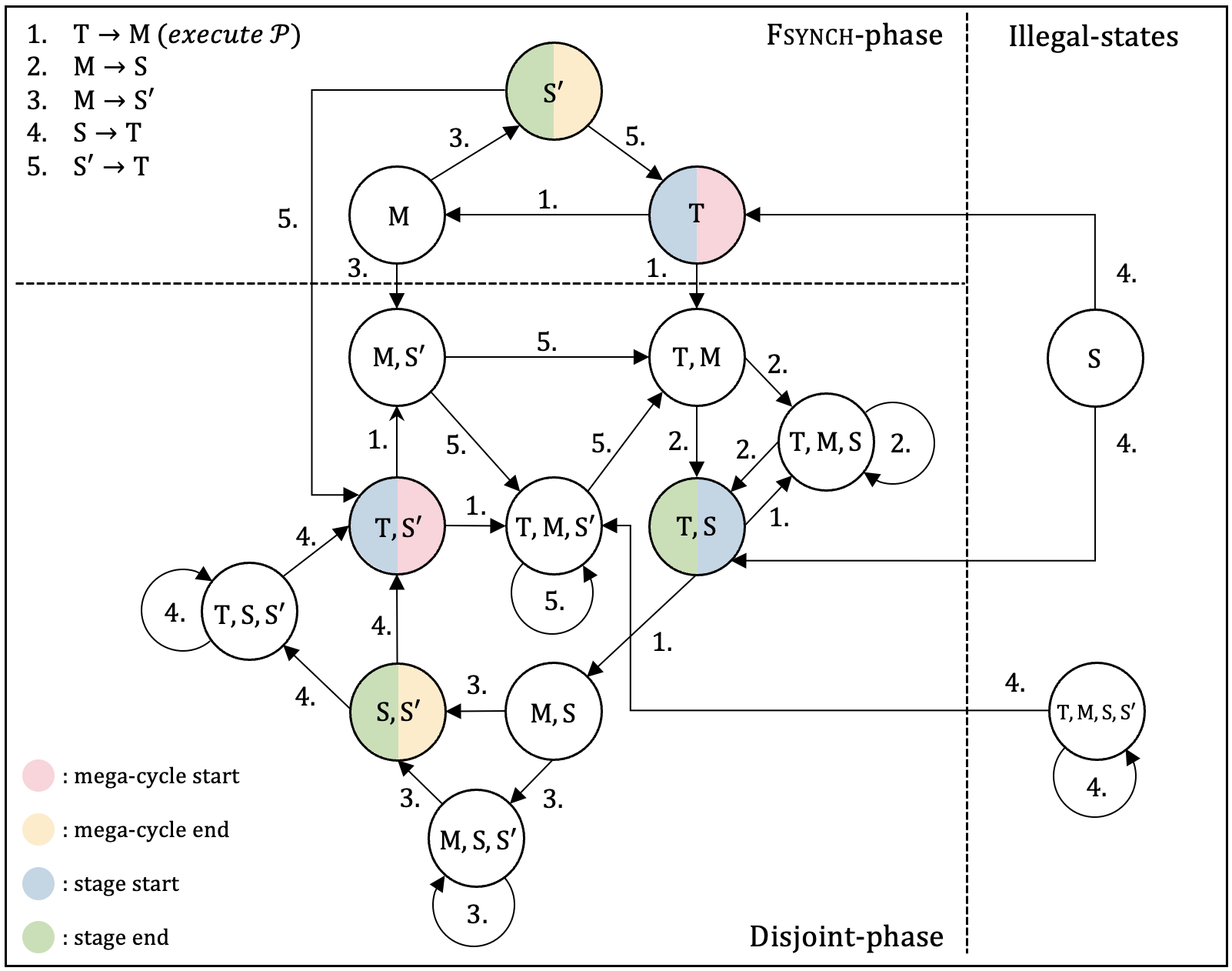}
\caption{Transition diagram of configurations (protocol $\textup{SIM}^{RS}_{S}$(the left part of this figure) and self-stabilizing ss-$\textup{SIM}^{RS}_{S}$).}\label{fig:alg4-rs-ss-statetran}
\end{center}
\end{figure}

\paragraph*{Correctness}
The left part of Fig.~\ref{fig:alg4-rs-ss-statetran} shows the transition diagram of configurations when performing $\textup{SIM}^{RS}_{S}$.
Each megacycle begins with $\forall T$ and ends with $\forall S'$ or $\forall S,S'$. If it ends with $\forall S'$, it means that all robots have executed the algorithm (FSYNCH-phase), while ending with $\forall S,S'$ represents all other cases (Disjoint-phase). The $S'$-robot signifies that it has executed at the end of this megacycle, and it is guaranteed not to execute at the beginning of the next megacycle. If a megacycle starts with $\forall T,S'$, it indicates the FSYNCH-phase has ended, and the Disjoint-phase is being executed, ending with $\forall S,S'$. Regarding stages, the FSYNCH-phase starts with $\forall T$ and ends with $\forall S'$ in one stage, while the Disjoint-phase begins with either $\forall T,S$ or $\forall T,S'$ and ends with either $\forall T,S$ or $\forall S,S'$. A stage beginning with $\forall T,S$ is any stage other than the first, while one starting with $\forall T,S'$ is the first stage. Furthermore, ending with $\forall T,S$ is any stage other than the last and ending with $\forall S,S'$ is the last stage. A stage that ends with $\forall S,S'$ transitions to $\forall T,S'$, and a new megacycle begins. The next robots to execute are chosen from those other than $S'$-robots, operating in a way that satisfies the disjointness of \RSY. \\

We now demonstrate that Protocol $\textup{SIM}^{RS}_{S}$ (Algorithm~\ref{alg4-rs-ss}) ensures fair and accurate execution of any \textit{restricted-repetition} protocol ${\cal{P}}$. This is achieved by defining \textit{mega-cycles} and \textit{stages} within the protocol, based on the global configuration of the robots.

A \textit{mega-cycle} starts at time $t$ when only robots in states $T$ and $S'$, denoted as $c(\in\{T, S'\})$, are present. The mega-cycle concludes at the earliest subsequent time $t'>t$, marked by the exclusive presence of robots in states $S$ and $S'$, denoted as $c'(\in\{S, S'\})$.




For transitions from the end of the one mega-cycle to the beginning of the next one, this is accomplished by the following transitions (Fig.~\ref{fig:alg4-rs-ss-statetran}):

\begin{itemize}
    \item \FSY-phase mega-cycle to \FSY-phase mega-cycle\\ 
    $\{S'\} \rightarrow \{T\}$ (the mega-cycle consists of one stage)\\
    \item \FSY-phase mega-cycle to Disjoint-phase mega-cycle\\
    $\{S'\} \rightarrow \{T,S'\}$\\
    \item Disjoint-phase mega-cycle to Disjoint-phase mega-cycle\\
    $\{S,S'\}\ (\rightarrow \{T,S,S'\}) \rightarrow \{T,S'\}$\\
\end{itemize}

When this process is completed, a new mega-cycle starts.\\

A $stage$ within a mega-cycle commences at time $t$ when only robots in states $c(\in\{T, S, S'\})$ are present, including at least one $T$-robot. During a stage, some robots transition to the $M$ state. This stage concludes at the earliest time $t'>t$ when only robots in states $c(\in\{T, S, S'\})$ are observed. The stage $R^i_j$ encompasses the set of $T$-robots that switch to M during the time interval $t_1\ (t_0 \leq t_1 < t_2)$.

Furthermore, consider $R^1,R^2,\ldots ,R^i,\ldots$ as a sequence of mega-cycles, with $R_1^i,R_2^i,\ldots ,R_{k_i}^i$ representing the sequence of stages within the $i$-th mega-cycle. Based on the \RSY\ scheduler, \RSY\ is categorized into two phases: the \FSY-phase, where all robots are activated simultaneously (i.e., $k_i=1$), and the Disjoint-phase, characterized by a lack of common robots activated across successive stages (i.e., $k_i \geq 2, R_j^i\cap R_{j+1}^i=\emptyset, R^i_{k_i} \cap R^{i+1}_1 = \emptyset$).


\begin{lemma}\label{lem:st-rs-s}
During each stage, the following conditions are met:
\begin{enumerate}
\item[(i)] At the beginning of the stage, there are one or more $T$-robots, and all other robots are either $S$ or $S'$ (denoted as $c(\in\{S, S'\})$-robots).
\item[(ii)] At least one $T$-robot executes protocol ${\cal{P}}$ during this stage.
\item[(iii)] At the end of the stage, those robots that executed ${\cal{P}}$ will transition to color $S'$ if no $T$-robots remain, otherwise, they transition to color $S$.
\item[(iv)] At the end of this stage, only robots in states $T$, $S$, or $S'$ (denoted as $c(\in\{T, S, S'\})$-robots) are present.
\end{enumerate}
\end{lemma}

\begin{proof}
Parts (i) and (iv) follow directly from the definition. $R^i_1$ is the earliest stage when only $c(\in\{T, S, S'\})$-robots exist (e.g at the beginning). 
During $R^i_1$, some of $T$-robots turn their color to M (and execute ${\cal{P}}$). Thereafter, all robots except $M$-robots, won't change the color (such a robot does not execute ${\cal{P}}$) in this stage. This proves parts (ii). According to the rules of the protocol, those robots turn their color to M and have executed one step of ${\cal{P}}$, would now change to color S' if $T$-robots are not exists in the configuration (at this time stage is $R^i_{k_i}$), otherwise change to S (stage is $R^i_j (1 \leq j < k_i)$). When all these robots' color turn to S or S', the $T$-robots would have a next chance to execute ${\cal{P}}$ and turn to M. This is the end of this stage and conditions (iii) and (iv) hold at this time.

In addition, from Fig.~\ref{fig:alg4-rs-ss-statetran}, the following transitions are made within each stage:

\begin{itemize}
    \item $R^i_1:$ At the beginning of this stage, none of the robots are execute ${\cal{P}}$.
    \begin{itemize}
        \item $\{T\} \rightarrow \{M\} \rightarrow \{M,S'\}\ (\rightarrow \{T,M,S'\}) \\\indent \ \rightarrow \{T,M\}\ (\rightarrow \{T,M,S\}) \rightarrow \{T,S\}$
        \item $\{T\} \rightarrow \{T,M\}\  (\rightarrow \{T,M,S\}) \rightarrow \{T,S\}$
        \item $\{T,S'\} \rightarrow \{T,M,S'\} \rightarrow \{T,M\}\ (\rightarrow \{T,M,S\}) \rightarrow \{T,S\}$
        \item $\{T,S'\} \rightarrow \{M,S'\}\ (\rightarrow \{T,M,S'\}) \\\indent \ \rightarrow \{T,M\}\ (\rightarrow \{T,M,S\}) \rightarrow \{T,S\}$
    \end{itemize}
    \item $R^i_j\ (2 \leq j < k_i):$ At the end of this stage, there are robots that have executed ${\cal{P}}$ once and robots that have never.
    \begin{itemize}
        \item $\{T,S\} \rightarrow \{T,M,S\} \rightarrow \{T,S\}$
    \end{itemize}
    \item $R^i_{k_i}:$ At the end of this stage, all robots have been executed ${\cal{P}}$ exactly once.
    \begin{itemize}
        \item $\{T\} \rightarrow \{M\} \rightarrow \{S'\}\ (k_i=1)$
        \item $\{T,S\} \rightarrow \{M,S\}\  (\rightarrow \{M,S,S'\}) \rightarrow \{S,S'\}\ (k_i\neq 1)$
    \end{itemize}
\end{itemize}
$\qed$
\end{proof}

\begin{lemma}\label{lem:mc-sim-rs-s-1}
For each mega-cycle, the following conditions are met:
\begin{enumerate}
\item[(i)] At the beginning of a mega-cycle, only robots in states $T$ and $S$, denoted as $c(\in\{T, S\})$-robots, are present.
\item[(ii)] $S'$-robots do not execute protocol ${\cal{P}}$ during the first stage of a mega-cycle.
\item[(iii)] Throughout the mega-cycle, each robot executes protocol ${\cal{P}}$ exactly once.
\item[(iv)] Each mega-cycle concludes within a finite time frame.
\item[(v)] At the end of a mega-cycle, only robots in states $S$ or $S'$, denoted as $c'(\in\{S, S'\})$-robots, are present.
\end{enumerate}
\end{lemma}

\begin{proof}
Due to Lemma~\ref{lem:st-rs-s} we know that, in each stage, a non-empty subset of the robots execute ${\cal{P}}$. Thus each stage of protocol $\textup{SIM}^{RS}_{S}$ corresponds to one activation round of an execution of ${\cal{P}}$ in the \RSY\ model where the set of robots activated in that round corresponds to the set of $T$-robots turn to color M, and eventually to $c\in\{S, S'\}$, during this stage of $\textup{SIM}^{RS}_{S}$. Since there are no $M$-robots exist at the stage $R^i_1$, $S'$-robots can be activated but cannot turn to color T. Also, at the stage $R^i_{k_i}$ $M$-robot exists, so $S'$-robots can turn to T, but cannot execute ${\cal{P}}$. This proves parts (ii). As long as there are more $T$-robots and $S'$-robots exist, another stage can begin. With each stage, the sum of those robots decrease. Eventually all robots will be colored $c\in\{S, S'\}$, and thus, each robot will have executed ${\cal{P}}$ exactly once (stage $R^i_{k_i}$). At the end of stage $R^i_{k_i}$, the mega-cycle ends and the conditions of the Lemma~\ref{lem:mc-sim-rs-s-1} are satisfied. $\qed$
\end{proof}

\begin{lemma}\label{lem:ms-sim-rs-s-2}
The mega-cycle of protocol $\textup{SIM}^{RS}_{S}$ corresponds to a synchronous activation and disjoint-semi-synchronous activation round in some execution of ${\cal{P}}$ in the \RSY\ model.
\end{lemma}

\begin{proof}
In each round of the \FSY-phase, all robots are activated simultaneously in one stage of the mega-cycle. For the Disjoint-phase, we derive the correspondence using the lemma that connects stages and mega-cycles. Specifically, $S'$-robots that executed ${\cal{P}}$ in the final stage $R_{k_i}^i$ will not execute ${\cal{P}}$ in the initial stage $R_1^{i+1}$ of the next mega-cycle.

We then demonstrate that the \FSY-phase transitions to the Disjoint-phase but does not revert to the \FSY-phase. The transition to the Disjoint-phase occurs when only a subset of robots is activated from a state where previously all robots had been activated simultaneously (initially all are T-robots; if all are activated together, they share the same state and time of executing ${\cal{P}}$). Conversely, the transition back to the \FSY-phase is precluded by the absence of a configuration shift from multiple states to a single state among the group of robots. This is supported by Fig.\ref{fig:alg4-rs-ss-statetran}, which illustrates that such a transition does not occur. $\qed$
\end{proof}

\begin{theorem}\label{th:SIM-RS-S}
Protocol $\textup{SIM}^{RS}_{S}$ is correct, i.e. any execution of protocol $\textup{SIM}^{RS}_{S}$ in \SSY\ corresponds to a possible execution of ${\cal{P}}$ in \RSY.
\end{theorem}
\begin{proof}
First notice that after each mega-cycle, only $c(\in\{S, S'\})$-robots exist (Lemma 2) and thus, according to the state transition diagram of protocol $\textup{SIM}^{RS}_{S}$, after each mega-cycle ends, the next mega-cycle begins automatically. Since each mega-cycle terminates in finite time, the execution of protocol $\textup{SIM}^{RS}_{S}$ is an infinite sequence of mega-cycles. We have already shown that each stage within a mega-cycle of protocol $\textup{SIM}^{RS}_{S}$ corresponds to a synchronous activation and disjoint-semi-synchronous activation round in some execution of ${\cal{P}}$ in the \RSY\ model (Lemma 3). We now need to show that the sequence of such activation rounds satisfies the fairness property. The fairness property requires that in any infinite execution of ${\cal{P}}$, each robot must be activated infinitely often. We have shown that in each mega-cycle, each robot actively executes ${\cal{P}}$ once. Thus in any infinite execution, i.e. an infinite sequence of mega-cycles, each robot executes ${\cal{P}}$ infinitely many times. Hence, this execution simulates a possible execution of protocol ${\cal{P}}$. In fact, this particular execution also satisfies the stronger condition of 1-fairness, where each robot is activated exactly once in a mega-cycle.

Note that if protocol ${\cal{P}}$ is a terminating algorithm (i.e., it terminates for every execution), then during the simulation, the robots would stop moving in finite time. $\qed$
\end{proof}

Thus, Theorem~\ref{4col-RS-SS} and its corollary hold.

\begin{corollary}
$\forall R \in \mathcal R,  {\LU_k^{RS}(R) \subseteq \LU_{4k}^{S}}(R).$
\end{corollary}

%% file: arxiv/4Col-self-stable.tex
An simulation protocol is {\em self-stabilizing} for protocol $\cal{P}$ if it satisfies the conditions of the scheduler under which $\cal{P}$ is executed from any initial configuration, stating with all robots in inactive. 

We can make the protocol $\textup{SIM}^{RS}_{S}$ self-stabilizing (denoted as ss-$\textup{SIM}^{RS}_{S}$).
Fig.~\ref{fig:alg4-rs-ss-tran} (b) shows ss-$\textup{SIM}^{RS}_{S}$, and the red-labeled part was added to achieve self-stabilization.

We can show that $\textup{SIM}^{RS}_{S}$ works correctly even if it starts from any configuration appearing on the left part of Fig.~\ref{fig:alg4-rs-ss-statetran}, and ss-$\textup{SIM}^{RS}_{S}$ works correctly from any configuration by adding the red-labeled part in Fig.~\ref{fig:alg4-rs-ss-tran} (b).

We show ss-$\textup{SIM}^{RS}_{S}$(Algorithm~\ref{fig:alg4-rs-ss}) is correct and self-stabilizing. First,  $\textup{SIM}^{RS}_{S}$ works correctly from any configuration of the \FSY-phase or the Disjoint-phase.

\begin{algorithm}[H]
    \caption{ss-$\textup{SIM}^{RS}_{S}$}\label{fig:alg4-rs-ss}
    State \textit{Look}
    \begin{algorithmic}[0]
        \STATE $Pos\lbrack{r}\rbrack:$ the position on the plane of robot $r$ (according to my coordinate system);
        \STATE $Light\lbrack{r}\rbrack:$ the color of the light of robot.
    \end{algorithmic}
    (Note: I am robot $x$) \\ \\
    State \textit{Compute}
    \begin{algorithmic}[1]
        \STATE $p \leftarrow Pos\lbrack{x}\rbrack$.
        \STATE $c \leftarrow \{Light\lbrack{r}\rbrack\ |\ \forall r \in R\}$.
        \IF {$Light\lbrack{x}\rbrack$ = T}
            \IF {$c = \{\textup{T}\} \vee c = \{\textup{T,S}\} \vee c = \{\textup{T,S'}\}$}
                \STATE Execute ${\cal{P}}.$
                \STATE $p \leftarrow computed\ destination$.
                \STATE $Light\lbrack{x}\rbrack \leftarrow$ M.
            \ENDIF
        \ELSIF {$Light\lbrack{x}\rbrack$ = M}
            \IF {$c\cap\{\textup{S'}\} = \emptyset \wedge \{\textup{T}\} \subseteq c$}
                \STATE $Light\lbrack{x}\rbrack \leftarrow$ S.
            \ENDIF
            \IF {$c\cap\{\textup{T,S'}\} = \emptyset \vee (\{\textup{T}\} \subsetneq c \wedge \{\textup{S}\} \subseteq c)$}
                \STATE $Light\lbrack{x}\rbrack \leftarrow$ S'.
            \ENDIF
        \ELSIF {$Light\lbrack{x}\rbrack$ = S} 
            \IF {$(c\cap\{\textup{M}\} = \emptyset \wedge \{\textup{S'}\} \subseteq c) \vee c = \{\textup{S}\} \vee \{\textup{T,M,S'}\} \subseteq c$}
                \STATE $Light\lbrack{x}\rbrack \leftarrow$ T.
            \ENDIF
        \ELSIF {$Light\lbrack{x}\rbrack$ = S'}
            \IF {$c\cap\{\textup{S,T}\} = \emptyset \vee (c\cap\{\textup{S}\} = \emptyset \wedge \{\textup{M}\} \subseteq c)$}
                \STATE $Light\lbrack{x}\rbrack \leftarrow$ T.
            \ENDIF
        \ENDIF
    \vspace{3mm}
    \end{algorithmic} 
    
    State \textit{Move} 
    \begin{algorithmic}[0]
        \STATE Move($p$).
    \end{algorithmic}
\end{algorithm}

\begin{lemma}\label{lem:ss-RS-SS}
Protocol $\textup{SIM}^{RS}_{S}$ is correct even if it starts from any configuration which in the \FSY-phase or the Disjoint-phase.
\end{lemma}

\begin{proof}
We have shown that the $\textup{SIM}^{RS}_{S}$ can correctly simulate all robots starting with $T$.
For simplicity, we consider protocols with a finite number of executions.
Considering that the protocol $\textup{SIM}^{RS}_{S}$ is a sequence of activation rounds and that \RSY\ can start from the Disjoint-phase, we can say that the suffix of the execution sequence is also an execution sequence that can be simulated correctly.

Thus, let $s$ be a certain configuration that appears in the \FSY-phase and the Disjoint-phase, $ar$ be the set of arbitrary activation round sequences that can occur in a $\textup{SIM}^{RS}_{S}$ execution, and $s'_i$ be the starting state of round $e_i$ of the activation round sequence.

Furthermore, let $ar'$ be the set of activation round sequences in $ar$ for which there exists a round $t$ satisfying $s=s'_t$. 

If the initial state of the global state is $s$, then the execution is the same as the suffix of any of the activation round sequences in $ar'$.
It is correct due to the correctness of the $\textup{SIM}^{RS}_{S}$ and the synchronization of the execution by \SSY.

Therefore, even if we let the $\textup{SIM}^{RS}_{S}$ run with any global configuration that appears in the \FSY-phase and the Disjoint-phase as the starting configuration, we can still correctly simulate the \RSY\ protocol.
$\qed$
\end{proof}

\begin{theorem}\label{th:ss-RS-S}
Protocol ss-$\textup{SIM}^{RS}_{S}$ is correct and self-stabilizing, i.e. from any initial global configuration, any execution of protocol $\textup{SIM}^{RS}_{S}$ in \SSY\ corresponds to a possible execution of ${\cal{P}}$ in \RSY.   
\end{theorem}

\begin{proof}
Regarding the added transition conditions, since they all perform transitions without executing $\cal{P}$, they only perform a color change from Illegal-states, and eventually merge into global state of \FSY-phase or Disjoint-phase.
Thus, it is the same as starting with the merging state as the initial state.

Therefore, from Lemma~\ref{lem:ss-RS-SS}, ss-$\textup{SIM}^{RS}_{S}$ can correctly simulate the \RSY\ protocol even if it starts from an arbitrary configuration. $\qed$
\end{proof}

%% file: arxiv/5Col-RS-AS.tex

If we use one more color (that is, use 5 colors), protocol $\textup{SIM}^{RS}_{S}$ (resp. ss-$\textup{SIM}^{RS}_{S}$) can be extended so that it works in \ASY\ to simulate \RSY\ from an initial configuration (resp. any initial configuration). 
They are called $\textup{SIM}^{RS}_{A}$ and ss-$\textup{SIM}^{RS}_{A}$, and shown in Fig.~\ref{fig:diagram-SIMRSA} (a) and (b), respectively.

In addition to $T$, $M$, $S$, $S'$ which have the same meaning as the colors in Protocol $\textup{SIM}^{RS}_{S}$, a color $W$ (Waiting) is introduced. A robot $r$ that executes the simulated algorithm changes from $T$ to $M$ and executes the algorithm. Unlike in \SSY, in \ASY, robots observing 
$r$ before it changes to $M$ get the same snapshot as $r$, but after changing to $M$, the snapshot differs, thus robots observing an $M$-robot cannot execute the algorithm. At this time, a robot observing an $M$-robot changes its color from $T$ to $W$ and pend its execution until the next stage. Afterward, once all $M$-robots have completed their algorithm execution and changed their color to $S$, the $W$-robots return to $T$, preparing for the execution of the next stage. 

These transitions, along with those involving the color of $W$, are the same as in Protocol $\textup{SIM}^{RS}_{S}$, except for the transitions related to $W$, where $M$-robots change their color from $M$ to $S$ or $S'$ if $W$-robots exist and $S'$-robots and $T$-robots do not exist, or $S$-robots exist and $T$-robots and $W$-robots do not exist, and $S$-robots or $S'$-robots change their color to $T$ if there do not exist $W$-robots\footnote{Based on the meanings of $S$ and $S'$, the transition from $S$ to $T$ occurs when there are no $W$-robots and $M$-robots, but there are $S'$-robots present. Conversely, the transition from $S'$ to $T$ occurs when there are no $W$-robots and $S$-robots, but $M$-robots are present.}. Moreover, after all robots have executed the algorithm (\FSY-phase), they all become $M$-robots. In this case, any robot observing all $M$-robots changes their color to $W$. These $W$-robots then become capable of executing the algorithm in the next stage.

The transition diagrams of configurations when performing $\textup{SIM}^{RS}_{A}$ and ss-$\textup{SIM}^{RS}_{A}$ are shown in Fig.~\ref{fig:diagram-trans-SIMRSA}. The correctness of the protocols can be shown in a way similar to the cases of $\textup{SIM}^{RS}_{S}$ and ss-$\textup{SIM}^{RS}_{S}$.

We show that asynchronous systems equipped with a light of five colors$(\LU_5^{A})$, are at least as powerful as \textit{restricted-repetition} system devoid of lights.

The robots in the system can exist in one of the following colors:
\begin{itemize} \item \textbf{T(rying)}: The robot has not executed ${\cal{P}}$ in the current mega-cycle. \item \textbf{W(aiting)}: The robot has no chance to execute ${\cal{P}}$ in the current stage. \item \textbf{M(oving)}: The robot has executed ${\cal{P}}$ once. \item \textbf{S(topped)}: The robot has executed ${\cal{P}}$ but is not in the last stage of the mega-cycle. \item \textbf{S'(topped)}: The robot has executed ${\cal{P}}$ in the last stage of the mega-cycle. \end{itemize}

\paragraph*{Protocol Description}

The pseudo code of the protocol are presented in Algorithm~\ref{fig:alg-sim-RS-A}, while the transition diagram showing the changes in robot colors is represented in Fig.~\ref{fig:diagram-SIMRSA}.

The light used in $\textup{SIM}^{RS}_{A}$ can display five colors: $T$, $M$, $S$, $S'$, $W$. Initially, all light are set to $T$.

We denote state set as a global configuration (e.g. $\forall T,S$ mean all robots state is either T or S, and both color robot exists at least one).

Fig.~\ref{fig:diagram-trans-SIMRSA} is a transition diagram of global configurations. These global configurations are important state to simulate \RSY.
\begin{itemize}
    \item $\forall T$: The first stage of each mega-cycle in \FSY-phase.
    \item $\forall T,S'$: The first stage of each mega-cycle in Disjoint-phase.
    \item $\forall M$: The last stage of each mega-cycle in \FSY-phase.
    \item $\forall S,S'$: The last stage of each mega-cycle in Disjoint-phase.
    \item $\forall T,S$: The stage which is not first or last of each mega-cycle in Disjoint-phase.
\end{itemize}


\begin{figure}[H]
\begin{center}
\includegraphics[width=120mm]{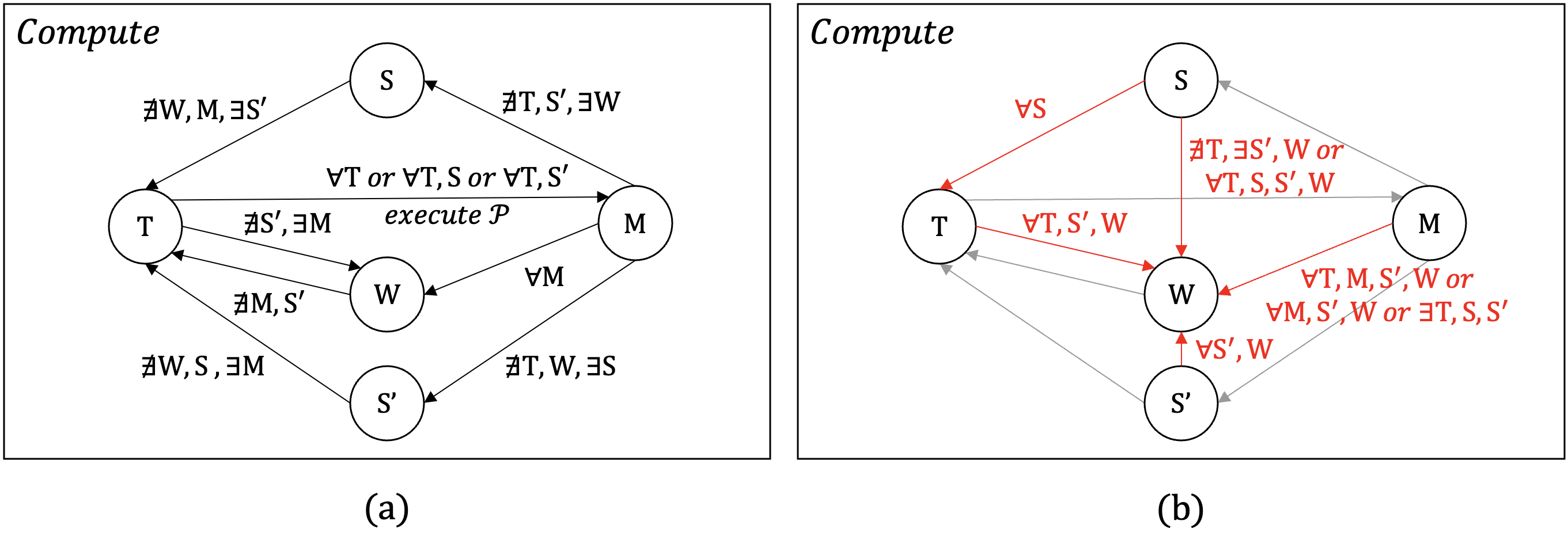}
\caption{Transition diagram representations of (a) protocol $\textup{SIM}^{RS}_{A}$ and (b) self-stabilizing protocol ss-$\textup{SIM}^{RS}_{A}$. The condition colored red in (b) is  newly added to $\textup{SIM}^{RS}_{A}$ to achieve self-stabilization, but not newly added in ss-$\textup{SIM}^{RS}_{A}$ are omitted.}\label{fig:diagram-SIMRSA}
\end{center}
\end{figure}

\begin{algorithm}[H]
    \caption{$\textup{SIM}^{RS}_{A}$}\label{fig:alg-sim-RS-A}
    State \textit{Look}
    \begin{algorithmic}[0]
        \STATE $Pos\lbrack{r}\rbrack:$ the position on the plane of robot $r$ (according to my coordinate system);
        \STATE $Light\lbrack{r}\rbrack:$ the color of the light of robot.
    \end{algorithmic}
    (Note: I am robot $x$) \\ \\
    State \textit{Compute}
    \begin{algorithmic}[1]
        \STATE $p \leftarrow Pos\lbrack{x}\rbrack$.
        \STATE $c \leftarrow \{Light\lbrack{r}\rbrack\ |\ \forall r \in R\}$.
        \IF {$Light\lbrack{x}\rbrack$ = T}
            \IF {$c = \{\textup{T}\} \vee c = \{\textup{T,S}\} \vee c = \{\textup{T,S'}\}$}
                \STATE Execute ${\cal{P}}.$
                \STATE $p \leftarrow computed\ destination$.
                \STATE $Light\lbrack{x}\rbrack \leftarrow$ M.
            \ENDIF
            \IF {$c\cap\{\textup{S'}\} = \emptyset \wedge \{\textup{M}\} \subseteq c$}
                \STATE $Light\lbrack{x}\rbrack \leftarrow$ W.
            \ENDIF
        \ELSIF {$Light\lbrack{x}\rbrack$ = M}
            \IF {$c = \{\textup{M}\}$}
                \STATE $Light\lbrack{x}\rbrack \leftarrow$ W.
            \ENDIF
            \IF {$c\cap\{\textup{T,S'}\} = \emptyset \wedge \{\textup{W}\} \subseteq c$}
                \STATE $Light\lbrack{x}\rbrack \leftarrow$ S.
            \ENDIF
            \IF {$c\cap\{\textup{T,W}\} = \emptyset \wedge \{\textup{S}\} \subseteq c$}
                \STATE $Light\lbrack{x}\rbrack \leftarrow$ S'.
            \ENDIF
        \ELSIF {$Light\lbrack{x}\rbrack$ = S}
            \IF {$c\cap\{\textup{W,M}\} = \emptyset \wedge \{\textup{S'}\} \subseteq c$}
                \STATE $Light\lbrack{x}\rbrack \leftarrow$ T.
            \ENDIF
        \ELSIF {$Light\lbrack{x}\rbrack$ = S'}
            \IF {$c\cap\{\textup{W,S}\} = \emptyset \wedge \{\textup{M}\} \subseteq c$}
                \STATE $Light\lbrack{x}\rbrack \leftarrow$ T.
            \ENDIF
        \ELSIF {$Light\lbrack{x}\rbrack$ = W}
            \IF {$c\cap\{\textup{M,S'}\} = \emptyset$}
                \STATE $Light\lbrack{x}\rbrack \leftarrow$ T.
            \ENDIF
        \ENDIF
        \vspace{3mm}
    \end{algorithmic} 
    
    State \textit{Move} 
    \begin{algorithmic}[0]
        \STATE Move($p$).
    \end{algorithmic}
\end{algorithm}



The protocol simulates a sequence of mega-cycles, each starting with all robots attempting to execute protocol ${\cal{P}}$ ($\forall T$ or $\forall T,S'$) and ending with all robots having executed ${\cal{P}}$ once ($\forall M$ or $\forall S,S'$). After this end configuration, the system transitions to the starting configuration, and a new mega-cycle begins.

During each mega-cycle, each robot has the opportunity to execute one step of protocol ${\cal{P}}$. The following describes how the robot in each state operates in each global configuration, including state transitions.
\begin{itemize}
    \item $T$: The robot has not executed ${\cal{P}}$ in the current mega-cycle.
        \begin{itemize}
            \item $T\rightarrow{M}$\\
                $\forall T,S$ and $\forall T,S'$ are the possible global configuration when this state transition occurs.\\
                As explained above, these are the first stages in the mega-cycle, $T$-robots are trying to execute one LCM cycle of protocol ${\cal{P}}$ in this state transition.\\
            \item $T\rightarrow{W}$\\
                $\nexists S', \exists M$ are the possible global configuration when this state transition occurs.\\
                At this time, executed ${\cal{P}}$ and potentially moving robots are exist, so that $T$-robots are not allowed to move and changes its color to $W$ and waits for the next turn (i.e., it waits until no $M$-robots exist in global configuration).\\
        \end{itemize}
    \item $W$: The robot has no chance to execute ${\cal{P}}$ in the current stage. 
        \begin{itemize}
            \item $W\rightarrow{T}$\\
                $\nexists M,S'$ are the possible global configuration when this state transition occurs.\\
                At this time, no robots executed ${\cal{P}}$ and potentially moving or such a robot may never have existed (i.e., global configuration is $\forall W$) in the same stage. In order to get another chance to execute ${\cal{P}}$ in the next stage or mega-cycle, all $W$-robots change their state to $T$.\\
        \end{itemize}
    \item $M$: The robot has executed ${\cal{P}}$ once. 
        \begin{itemize}
            \item $M\rightarrow{W}$\\
                $\forall M$ is the only possible global configuration when this state transition occurs.\\
                This occurs in \FSY-phase when all $T$-robots execute protocol ${\cal{P}}$ in the same single stage and mega-cycle ended.\\
            \item $M\rightarrow{S}$\\
                $\nexists T,S', \exists W$ are the possible global configuration when this state transition occurs.\\
                From global configuration, some robots could not executed ${\cal{P}}$ in the same stage. This means $M$-robots executed ${\cal{P}}$ in Disjoint-phase but not in the last stage of the mega-cycle.\\
            \item $M\rightarrow{S'}$\\
                $\nexists T,W, \exists S$ are the possible global configuration when this state transition occurs.\\
                From global configuration, no robot are pending execution ${\cal{P}}$ in the same stage. This means $M$-robots executed ${\cal{P}}$ in Disjoint-phase and in the last stage of the mega-cycle.\\
        \end{itemize}
    \item $S$: The robot has executed ${\cal{P}}$ but is not in the last stage of the mega-cycle. 
        \begin{itemize}
            \item $S\rightarrow{T}$\\
                $\nexists W,M, \exists S'$ are the possible global configuration when this state transition occurs.\\
                From global configuration, all robots had executed ${\cal{P}}$ once in the current stage. As the mega-cycle ends, $S$-robot has the right to execute protocol ${\cal{P}}$ in the first stage of the next mega-cycle.\\
        \end{itemize}
    \item $S'$: The robot has executed ${\cal{P}}$ in the last stage of the mega-cycle.
        \begin{itemize}
            \item $S'\rightarrow{T}$\\
                $\nexists W,S, \exists M$ are the possible global configuration when this state transition occurs.\\
                From global configuration, all $S$-robots in the previous stage changed their state to $T$ and after some them are already executed ${\cal{P}}$ in current stage. This means no $T$-robots can get chance to execute ${\cal{P}}$ in the same stage, so that $S'$-robots can change their state to $T$ and wait for next stage. \\
        \end{itemize}
\end{itemize}

\begin{figure}[H]
\begin{center}
\includegraphics[width=120mm]{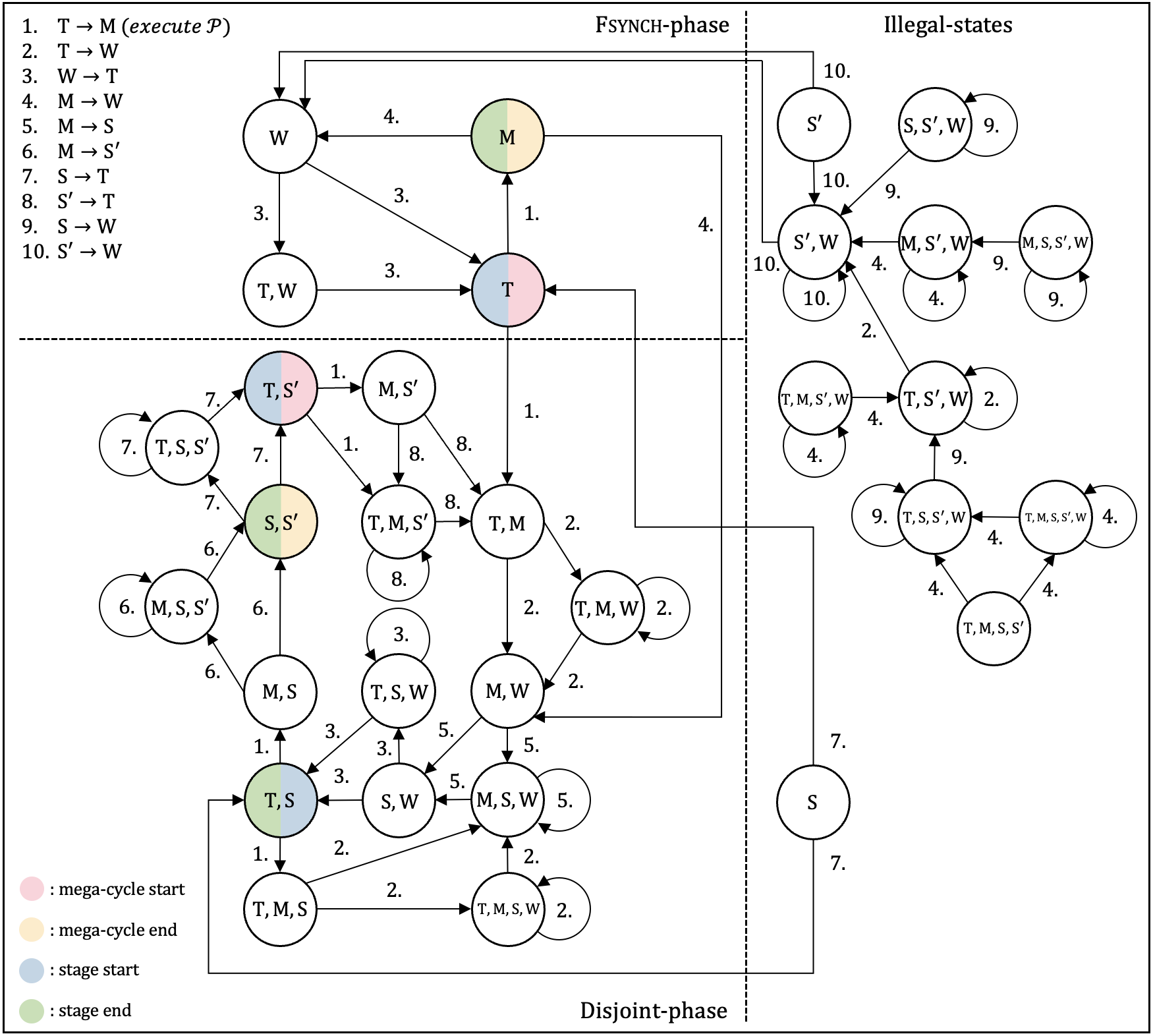}
\caption{Transition diagram of configurations (protocol $\textup{SIM}^{RS}_{A}$(the left part of this figure) and self-stabilizing ss-$\textup{SIM}^{RS}_{A}$).}\label{fig:diagram-trans-SIMRSA}
\end{center}
\end{figure}

\paragraph*{Correctness}

We now prove that Protocol $\textup{SIM}^{RS}_{A}$ provides a fair and correct execution of any \textit{restricted-repetition} protocol ${\cal{P}}$. We first define the concepts of \textit{mega-cycles} and \textit{stages} in this protocol, in terms of the global configuration of the robots.

A \textit{mega-cycle} begins at the time instant $t$ when only $c(\in\{T, S'\})$-robots exist, and mega-cycle ends at the earliest time instant $t'>t$ when only M-robots or only $c'(\in\{S, S'\})$-robots exist. 

For transitions from one mega-cycle to the next, this is accomplished by the following transitions (Fig.~\ref{fig:diagram-trans-SIMRSA}):
\begin{itemize}
    \item \FSY-phase mega-cycle to \FSY-phase mega-cycle\\ 
    $\{M\} \rightarrow \{W\}\ (\rightarrow \{T,W\}) \rightarrow \{T\}$\\
    \item \FSY-phase mega-cycle to Disjoint-phase mega-cycle\\
    $\{M\} \rightarrow \{M,W\}\ (\rightarrow \{M,S,W\}) \rightarrow \{S,W\}\ (\rightarrow \{T,S,W\}) \rightarrow \{T,S\}\ \\\indent \ (\rightarrow \{T,M,S\}\ (\rightarrow \{T,M,S,W\})\ \rightarrow \{M,S,W\} \rightarrow \ldots) \rightarrow \{M,S\}\ \\\indent \ (\rightarrow \{M,S,S'\}) \rightarrow \{S,S'\}\ (\rightarrow \{T,S,S'\}) \rightarrow \{T,S'\}$\\
    \item Disjoint-phase mega-cycle to Disjoint-phase mega-cycle\\
    $\{S,S'\}\ (\rightarrow \{T,S,S'\}) \rightarrow \{T,S'\}$
\end{itemize}

When each of these process is completed, a new mega-cycle starts.\\

A $stage$ of a mega-cycle begins at time $t_0$ when only $c(\in\{T, S, S'\})$-robots exist (with $T$-robot exists). During the stage, some robots change to color M and the stage ends at the earliest subsequent time $t_2>t_0$ when only $c(\in\{T, S, S'\})$-robots or $M$-robots exist. stage $R^i_j$ is the set of $T$-robots which turn to M in time $t_1\ (t_0 \leq t_1 < t_2)$.

The proof about the mega-cycle simulating the \RSY\ scheduler can be derived almost identically to Lemmas~\ref{lem:mc-sim-rs-s-1} and~\ref{lem:ms-sim-rs-s-2}. The only difference is the addition of the case where all robots are M in the mega-cycle's terminate condition. Therefore, we prove here that the behavior of each stage reproduces the activation rounds in \RSY.

\begin{lemma}
For each stage of the protocol, the following conditions are met:
\begin{enumerate}
\item[(i)] At the beginning of the stage, there are one or more $T$-robots, and all other robots are either in states $S$ or $S'$ (denoted as $c(\in\{S, S'\})$-robots).
\item[(ii)] At least one $T$-robot executes protocol ${\cal{P}}$ during this stage.
\item[(iii)] All robots that execute ${\cal{P}}$ during this stage will have the same snapshot.
\item[(iv)] At the conclusion of this stage, only robots in states $T$, $S$, $S'$ (denoted as $c'(\in\{T, S, S'\})$), or M-robots are present
\end{enumerate}
\end{lemma}

\begin{proof}
Parts (i) and (iv) follow directly from the definition. Let $t_0$ be the earliest $j$-th stage's time when only $c(\in\{T, S, S'\})$-robots exist (e.g at the beginning), and let ${\cal{S}}$ be the snapshot at that time. Let $t_2>t_0$ be the first time since $t_0$ when $T$-robot turn its color to M (according to the rules of $\textup{SIM}^{RS}_{A}$, this is the only possibility when only $c(\in\{T, S\})$-robots or only $c'(\in\{T, S'\})$-robots exist). Notice that this robot must have performed the \textit{Look} operation at time $t_1$ such that $t_0 \leq t_1 < t_2$. Since during the period $(t_0,t_2)$ the global configuration does not change, every robot performs \textit{Look} between $t_0$ and $t_2$ has the same snapshot ${\cal{S}}$, and in this snapshot all robots are colored $c\in\{T, S\}$ or $c'\in\{T, S'\}$. So, all $T$-robots would eventually execute ${\cal{P}}$ and turn to M in the stage $R^i_j$. Any robot $r$ that performs\textit{Look} at a time $\geq t_2$, would see some $M$-robots and thus $r$ would turn to W (such a robot does not execute ${\cal{P}}$ in stage $R^i_j$). This proves parts (ii) and (iii). 

According to the rules of the protocol, from initial configuration (i.e only $T$-robots exist), if all robot execute ${\cal{P}}$ and turn their color to M simultaneously, the stage $R^i_j$ ends immediately. If not so, those $M$-robots and have executed one step of ${\cal{P}}$, would now change to color S' if $c(\in\{T, W\})$-robots aren't exists and $S$-robot exists in the configuration (at this time stage is $R^i_{k_i}$), otherwise turn to S (stage is $R^i_j\ (1 \leq j < k_i)$). When all these robots' color turn to S or S', the $T$-robots would have a next chance to execute ${\cal{P}}$ and turn to M. This is the end of stage $R^i_j$ and (iv) hold at this time.

In addition, from Fig.~\ref{fig:diagram-trans-SIMRSA}, the following transitions are made within each stage:

\begin{itemize}
   \item $R^i_1:$ At the beginning of this stage, none of the robots are execute ${\cal{P}}$.
   \begin{itemize}
       \item $\{T\} \rightarrow \{T,M\} \rightarrow \{T,M,S'\}\ (\rightarrow \{T,M,W\}) \rightarrow \{M,W\}\ \\\indent \ (\rightarrow \{M,S,W\}) \rightarrow \{S,W\}\ (\rightarrow \{T,S,W\}) \rightarrow \{T,S\}$
       \item $\{T,S'\} \rightarrow \{T,M,S'\} \rightarrow \{T,M\}\ (\rightarrow \{T,M,W\}) \rightarrow \{M,W\}\ \\\indent \ (\rightarrow \{M,S,W\}) \rightarrow \{S,W\}\ (\rightarrow \{T,S,W\}) \rightarrow \{T,S\}$
       \item $\{T,S'\} \rightarrow \{M,S'\}\ (\rightarrow \{T,M,S'\}) \rightarrow \{T,M\}\  (\rightarrow \{T,M,W\})\ \\\indent \rightarrow \{M,W\}\ (\rightarrow \{M,S,W\}) \rightarrow \{S,W\}\ (\rightarrow \{T,S,W\}) \rightarrow \{T,S\}$
   \end{itemize}
   \item $R^i_j\ (2 \leq j < k_i):$ At the end of this stage, there are robots that have executed ${\cal{P}}$ once and robots that have never.
   \begin{itemize}
       \item $\{T,S\} \rightarrow \{T,M,S\}\ (\rightarrow \{T,M,S,W\}) \rightarrow \{M,S,W\} \\\indent \ \rightarrow \{S,W\}\ (\rightarrow \{T,S,W\}) \rightarrow \{T,S\}$
   \end{itemize}
   \item $R^i_{k_i}:$ At the end of this stage, all robots have been executed ${\cal{P}}$ exactly once.
   \begin{itemize}
       \item $\{T\} \rightarrow \{M\}\ (k_i=1)$
       \item $\{T,S\} \rightarrow \{M,S\}\  (\rightarrow \{M,S,S'\}) \rightarrow \{S,S'\}\ (k_i\neq 1)$
   \end{itemize}
\end{itemize} $\qed$
\end{proof}

\begin{theorem}
Protocol $\textup{SIM}^{RS}_{A}$ is correct, i.e. any execution of protocol $\textup{SIM}^{RS}_{A}$ in \ASY $^{5}$ corresponds to a possible execution of ${\cal{P}}$ in \RSY.
\end{theorem}

\begin{proof}
First notice that after each mega-cycle, only $M$-robots or only $c(\in\{S, S'\})$-robots exist and thus, according to the state transition diagram of protocol $\textup{SIM}^{RS}_{A}$, after each mega-cycle ends, the next mega-cycle begins automatically. Since each mega-cycle terminates in finite time, the execution of protocol $\textup{SIM}^{RS}_{A}$ is an infinite sequence of mega-cycles. We could say that each stage within a mega-cycle of protocol $\textup{SIM}^{RS}_{A}$ corresponds to a synchronous activation and disjoint-semi-synchronous activation round in some execution of ${\cal{P}}$ in the \RSY\ model. We now need to show that the sequence of such activation rounds satisfies the fairness property. The fairness property requires that in any infinite execution of ${\cal{P}}$, each robot must be activated infinitely often. We have shown that in each mega-cycle, each robot actively executes ${\cal{P}}$ once. Thus in any infinite execution, i.e. an infinite sequence of mega-cycles, each robot executes ${\cal{P}}$ infinitely many times. Hence, this execution simulates a possible execution of protocol ${\cal{P}}$. In fact, this particular execution also satisfies the stronger condition of 1-fairness.

Note that if protocol ${\cal{P}}$ is a terminating algorithm (i.e., it terminates for every execution), then during the simulation, the robots would stop moving in finite time.
 $\qed$
\end{proof}

Thus, we have the following theorem and corollary.
\begin{theorem}\label{5col-RS-AS}
$\forall R \in {\cal{R}}, \OB^{RS}(R) \subseteq \LU_{5}^{A}(R).$
\end{theorem}

\begin{corollary}
$\forall R \in \mathcal R,  \LU_k^{RS}(R) \subseteq \LU_{5k}^{A}(R).$
\end{corollary}

%% file: arxiv/5Col-self-stable.tex
We show that the protocol ss-$\textup{SIM}^{RS}_{A}$ (Algorithm~\ref{fig:ss-sim-rs-a}) correctly works and self-stabilizing. 
Although the proof is the same as ss-$\textup{SIM}^{RS}_{A}$, due to no bound on delays, \ASY\ scheduler cannot divide into rounds. Thus, we have to consider the time which can delimited by the global configuration. 

\begin{lemma}
Protocol $\textup{SIM}^{RS}_{A}$ is correct even if it starts from any configuration which in the \FSY-phase or the Disjoint-phase.
\end{lemma}

\begin{proof}
We showed that the $\textup{SIM}^{RS}_{A}$ can correctly simulate all robots starting with T.
For simplicity, we consider protocols with a finite number of executions.
In an \ASY\ scheduler, the focus is on the activation sequence after the time of the global state change.
The global state can change at the time when the computation in the LCM-cycle is completed. We denote such a time by $t_{CE}$.

Starting with the global state at time t as the initial state, as opposed to starting with all robots initially stopped, is inadequate because it does not assume that Move in the LCM-cycle will be executed.
This can be accomplished by making the robot that was started at time $t_{CE}$ behave as if it will do nothing when the LCM-cycle is executed.

Hence, as in the proof in Lemma~\ref{lem:ss-RS-SS}, using the notion that suffix of the execution sequence is also an execution sequence that can be simulated correctly.

Therefore, even if we let the $\textup{SIM}^{RS}_{A}$ run with any global state that appears in the \FSY-phase and the Disjoint-phase as the starting state, we can still correctly simulate the \RSY\ protocol.
 $\qed$
\end{proof}

\begin{algorithm}[H]
    \caption{ss-$\textup{SIM}^{RS}_{A}$}\label{fig:ss-sim-rs-a}
    State \textit{Look}
    \begin{algorithmic}[0]
        \STATE $Pos\lbrack{r}\rbrack:$ the position on the plane of robot $r$ (according to my coordinate system);
        \STATE $Light\lbrack{r}\rbrack:$ the color of the light of robot.
    \end{algorithmic}
    (Note: I am robot $x$) \\ \\
    State \textit{Compute}
    \begin{algorithmic}[1]
        \STATE $p \leftarrow Pos\lbrack{x}\rbrack$.
        \STATE $c \leftarrow \{Light\lbrack{r}\rbrack\ |\ \forall r \in R\}$.
        \IF {$Light\lbrack{x}\rbrack$ = T}
            \IF {$c = \{\textup{T}\} \vee c = \{\textup{T,S}\} \vee c = \{\textup{T,S'}\}$}
                \STATE Execute ${\cal{P}}.$
                \STATE $p \leftarrow computed\ destination$.
                \STATE $Light\lbrack{x}\rbrack \leftarrow$ M.
            \ENDIF
            \IF {$(c\cap\{\textup{S'}\} = \emptyset \wedge \{\textup{M}\} \subseteq c) \vee c = \{\textup{T,S',W}\}$}
                \STATE $Light\lbrack{x}\rbrack \leftarrow$ W.
            \ENDIF
        \ELSIF {$Light\lbrack{x}\rbrack$ = M}
            \IF {$c = \{\textup{M}\} \vee c = \{\textup{M,S',W}\} \vee c = \{\textup{T,M,S',W}\} \vee \{\textup{T,S,S'}\} \subseteq c$}
                \STATE $Light\lbrack{x}\rbrack \leftarrow$ W.
            \ENDIF
            \IF {$c\cap\{\textup{T,S'}\} = \emptyset \wedge \{\textup{W}\} \subseteq c$}
                \STATE $Light\lbrack{x}\rbrack \leftarrow$ S.
            \ENDIF
            \IF {$c\cap\{\textup{T,W}\} = \emptyset \wedge \{\textup{S}\} \subseteq c$}
                \STATE $Light\lbrack{x}\rbrack \leftarrow$ S'.
            \ENDIF
        \ELSIF {$Light\lbrack{x}\rbrack$ = S}
            \IF {$(c\cap\{\textup{W,M}\} = \emptyset \wedge \{\textup{S'}\} \subseteq c) \vee c = \{\textup{S}\}$}
                \STATE $Light\lbrack{x}\rbrack \leftarrow$ T.
            \ENDIF
            \IF {$(c\cap\{\textup{T}\} = \emptyset \wedge \{\textup{S',W}\} \subseteq c) \vee c = \{\textup{T,S,S',W}\}$}
                \STATE $Light\lbrack{x}\rbrack \leftarrow$ W.
            \ENDIF
        \ELSIF {$Light\lbrack{x}\rbrack$ = S'}
            \IF {$c\cap\{\textup{W,S}\} = \emptyset \wedge \{\textup{M}\} \subseteq c$}
                \STATE $Light\lbrack{x}\rbrack \leftarrow$ T.
            \ENDIF
            \IF {$c = \{\textup{S',W}\}$}
                \STATE $Light\lbrack{x}\rbrack \leftarrow$ W.
            \ENDIF
        \ELSIF {$Light\lbrack{x}\rbrack$ = W}
            \IF {$c\cap\{\textup{M,S'}\} = \emptyset$}
                \STATE $Light\lbrack{x}\rbrack \leftarrow$ T.
            \ENDIF
        \ENDIF
        \vspace{3mm}
    \end{algorithmic} 
    
    State \textit{Move} 
    \begin{algorithmic}[0]
        \STATE Move($p$).
    \end{algorithmic}
\end{algorithm}

As with $\textup{SIM}^{RS}_{S}$, self-stabilizability can be achieved for $\textup{SIM}^{RS}_{A}$ by adding appropriate transition conditions. The transition diagram of global configurations is represented in the whole part of Fig.~\ref{fig:diagram-trans-SIMRSA}

\begin{theorem}
Protocol ss-$\textup{SIM}^{RS}_{A}$ is correct and self-stabilizing, i.e. any initial global configurations and any execution of protocol $\textup{SIM}^{RS}_{A}$ in \ASY $^{5}$ corresponds to a possible execution of ${\cal{P}}$ in \RSY.
\end{theorem}

\begin{proof}
Regarding the added transition conditions, since they all perform transitions without executing $\cal{P}$, they only perform a color change from Illegal-states, and eventually merge into global state of \FSY-phase or Disjoint-phase.
Thus, it is the same as starting with the merging state as the initial state.

Therefore, from Lemma~\ref{lem:ss-RS-SS}, the self-stabilizing $\textup{SIM}^{RS}_{A}$ can correctly simulate the \RSY\ protocol even if it starts from an arbitrary state. $\qed$
\end{proof}

%% file: arxiv/3col-RS-AS-2robots.tex
\subsection{3-Color Simulation of \RSY\ by \ASY}

In this subsection, we show the following theorem constructively.

\begin{theorem}\label{th:Sim-RS-by-AS-2robots}
$\forall R \in \mathcal R_2,  {\OB^{RS}}(R) \subseteq \LU_3^{A}(R).$
\end{theorem}

To do so, we present a $\LU_3^{A}$ protocol SIM-$2^{RS}_A$ that produces \RSY\ execution of any $\OB^{RS}$ protocol $\cal{P}$. We also show that the number of colors used in SIM-$2^{RS}_A$ 
is optimal.


The transition diagram representation is shown in Fig.~\ref{fig:SIM2}-(a) and the pseudo code of the protocol 
is shown in Algorithm~\ref{fig:sim-2-rs-a}. 
It uses three colors $T, M,$ and $S$. The meaning of the colors is almost the same as those of Protocol SIM$^{RS}_A$. The transition of configurations is shown in Fig.~\ref{fig:SIM2}-(b).
If SIM-$2^{RS}_A$ works in \SSY, it is easily verified that it simulates $\cal{P}$ working in \RSY\ as follows;
As long as the both robots continue to be activated simultaneously, since the transition repeats $(T,T) \rightarrow (M,M) \rightarrow (S,S) \rightarrow (T,T)$, and $a$ and $b$ have performed $\cal{P}$ simultaneously when changing $T$ to $M$,
SIM-$2^{RS}_A$ makes $\cal{P}$ work in \FSY. Once only one robot, say $a$, is activated at $(T,T)$\footnote{When it is activated at $(M,M)$ or $(S,S)$, we can show similarly noting that at any time, when only one robot is activated, two robots have executed \cal{P} simultaneously.}, the configuration becomes $(M,T)$ and $a$ has executed $\cal{P}$.
After that the transition of the configurations repeats the loop of
$(M,T) \rightarrow (S,T) \rightarrow (S,M) \rightarrow (T,M) \rightarrow (T,S) \rightarrow (M,S) \rightarrow (M,T)$ shown in Fig.~\ref{fig:SIM2}-(b) for any activation schedule in \SSY, where $b$ performs $\cal{P}$ and then $a$ performs $\cal{P}$ in the loop.
For example, when the configuration is $(M,T)$, 
since $a$ changes $M$ to $S$ and $b$ does not change its color in SIM-$2^{RS}_A$, if $a$ is activated once regardless of activation of $b$, the configuration becomes $(S,T)$. Other transitions are similar.
Thus these activation satisfies \RSY\ adding the preceding simultaneous activation of the both robots.
We show that SIM-$2^{RS}_A$ can work correctly in \ASY, although it is very complicated due to asynchronicity. 


\begin{figure}[H]
\begin{center}
\includegraphics[width=120mm]{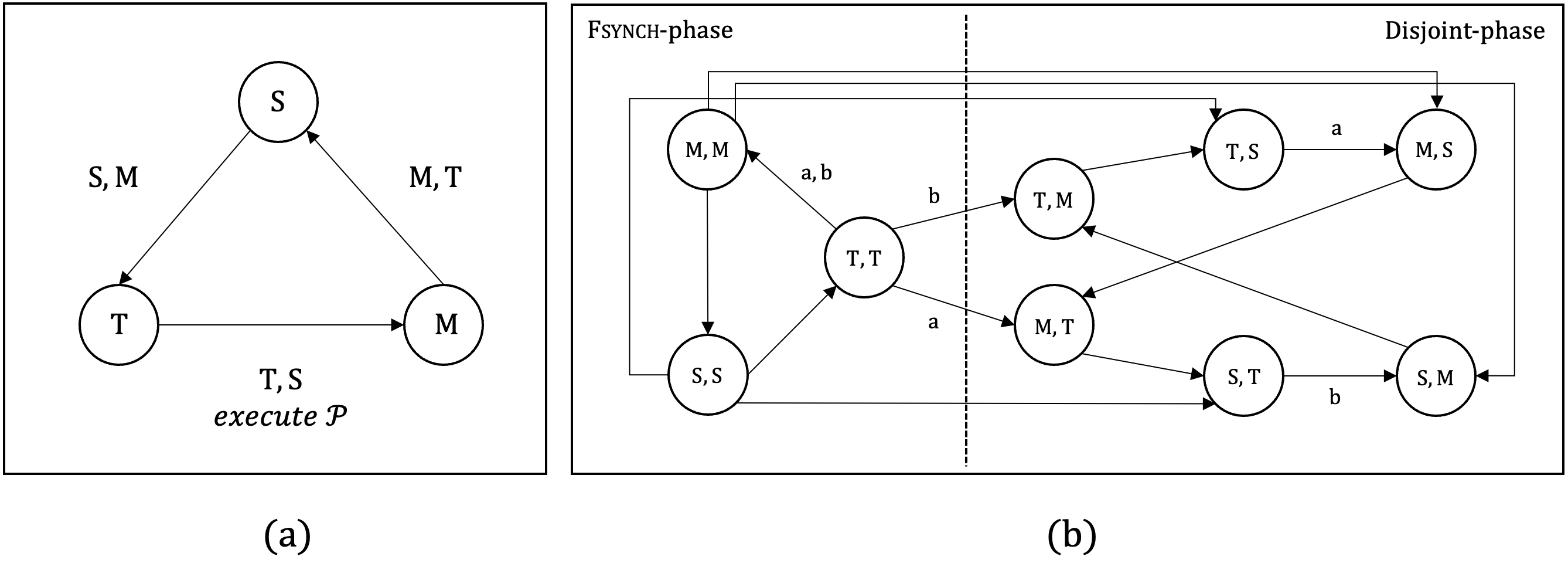}
\caption{(a) Transition diagram of protocol SIM-$2^{RS}_A$. (b) Transition diagram of state transitions (protocol SIM-$2^{RS}_A$)}\label{fig:SIM2}
\end{center}
\end{figure}

\begin{algorithm}[H]
    
    \caption{SIM-2$^{RS}_A$: for robot x at location x.pos}\label{fig:sim-2-rs-a}
   
    State \textit{Look} 
    \begin{algorithmic}[0] 
    
    \STATE $ my.light $
    \STATE $ other.light:$ the other robot's light.
    \vspace{3mm}
    \end{algorithmic}
    
    State \textit{Compute} 
    \begin{algorithmic}[1]
    \STATE $my.des \leftarrow my.pos$
    \IF {$my.light = T$ {\bf and} $other.light \in \{T,S\}$}
        \STATE Execute ${\cal{P}}.$
        \STATE $my.des \leftarrow computed\, destination$
        \STATE $my.light \leftarrow M$
    \ELSIF {$my.light = M$ {\bf and} $other.light \in \{T,M\}$}
        \STATE $my.light \leftarrow S$
        
    \ELSIF {$my.light = S$ {\bf and} $other.light \in \{M,S\}$}
        \STATE $my.light \leftarrow T$
    
    \ENDIF
    \vspace{3mm}
    \end{algorithmic}
    
    State \textit{Move} 
    \begin{algorithmic}[0]
        \STATE Move to $my.des$.
    \end{algorithmic}

\end{algorithm}




The following lemma holds for SIM-2$^{RS}_A$ (Algorithm~\ref{fig:sim-2-rs-a}). 

\begin{lemma}\label{lem:nochangecolor}
Let $(\alpha,\beta) \in \{(T,M), (S,T), (M,S)\}$. If $\alpha$-robot observes $\beta$-robot at time $t$ in SIM$^{RS}_A$,
the $\alpha$-robot does not change its color at $t'+1$ and until the next activation, where $t'$ is the time the $\alpha$-robot performs the $Comp$-operation.

\end{lemma}

In what follows, two robots are denoted as $a$ and $b$.
Given a robot $r$, an operation $op$($\in \{ \mathit{Look, Comp, M_{B}, M_{E}} \}$), and a time $t$,
$t^+(r, op)$ denotes the time $r$ performs the first $op$ after $t$ (inclusive) if there exists such operation, and
$t^-(r, op)$ denotes the time $r$ performs the first $op$ before $t$ (inclusive)  if there exists such operation.
If $t$ is the time the algorithm terminates, $t^+(r, op)$ is not defined for any $op$.
When $r$ does not perform $op$ before $t$ and $t^-(r, op)$ does not exist, $t^-(r, op)$ is defined to be $t_0$.


A time $t_c$ is called a \emph{cycle start time} (\emph{cs-time}, for short), if the next performed operations of both $a$ and $b$ after $t_c$ are both $\mathit{Look}$, or
otherwise, the robots performing the operations neither change their colors of lights nor move.
In the latter case, we can consider that these operations can be performed before $t_c$ and the subsequent $Look$ operation can be performed as the first operation after $t_c$.

If a configuration at a $cs$-time $t$ is that robot $a$ has color $\alpha$ and robot $b$ has color $\beta$, it is denoted as $[(\alpha,\beta)]_t$. If the algorithm changes configuration from $[(\alpha,\beta)]_t$ to  $[(\alpha',\beta')]_{t'}$ for some $cs$-time $t'>t$, it is denoted as $[(\alpha,\beta)]_t \rightarrow [(\alpha',\beta')]_{t'}$. 

If we write $[(\alpha,\beta)]_t \stackrel{\mathrm{\gamma}}{\mathrm{\rightarrow}} [(\alpha',\beta')]_{t'}$,
$\gamma$ denotes robot(s) executing $\cal{P}$ between $t$ and $t'$, where $\gamma=a$, $\gamma=b$, or $\gamma=(a,b)$ shows $a$ executes $\cal{P}$, $b$ executes $\cal{P}$, or $a$ and $b$ execute $\cal{P}$ simultaneously, respectively. Also if some robot, say $a$ executes $\cal{P}$ after $a$ and $b$ execute $\cal{P}$ simultaneously between $t$ and $t'$, $\gamma$ is written as $(a,b), a$.



\begin{lemma}\label{lem:ST->}
   Assume that an $M$-robot, say, $a$ performs $Comp$ and changes its color to $S$ at time $t$, and the other robot, say $b$ is an $S$-robot at $t+1$ and performs $Comp$ and changes its color to $T$ at $t' > t$\footnote{Note that $b$ observes $M$-robot or $S$-robot by $t'-1$.} (Figure~\ref{fig:ST->}).
   If $a$ is activated and performs $Look$ at $t'' \in (t..t']$, there exists a cs-time $\hat{t}$ such  that
   exactly one of the followings holds,
$[(S,S)]_{t+1} \stackrel{}{\mathrm{\rightarrow}} [(T,T)]_{\hat{t}}$, $[(S,S)]_{t+1} \stackrel{\mathrm{a}}{\mathrm{\rightarrow}}[(M,T)]_{\hat{t}}$, or $[(S,S)]_{t+1} \stackrel{\mathrm{b}}{\mathrm{\rightarrow}}[(T,M)]_{\hat{t}}$. 
\end{lemma}
\begin{proof}
There are two cases; (1) $t^+(a,C) \leq t'$, and (2) $t' <t^+(a,C)$.

{\bf Case (1)}: If $a$ is not activated between $t^+(a,C)+1$ and $t'$ after changing its color to $T$ at $t^+(a,C)$, ${\hat{t}}=t'+1$  is a $cs$-time such that $[(S,S)]_{t+1} \stackrel{}{\mathrm{\rightarrow}} [(T,T)]_{\hat{t}}$.

If $a$ is activated between $t^+(a,C)+1$ and $t'$ once and finishes its cycle 
 between  
$t^+(a,C)+1$ and $t'$, then $a$ changes its color to $M$ and executes $\cal{P}$. Even if $a$ is activated again until $t'$,
since $M$-robot $a$ observes $S$-robot $b$, $a$ does not change its color, 
${\hat{t}}=t'+1$ becomes a $cs$-time such that $[(S,S)]_{t+1} \stackrel{\mathrm{a}}{\mathrm{\rightarrow}} [(M,T)]_{\hat{t}}$.

If $a$ is activated between $t^+(a,C)+1$ and $t'$ but does not finish its cycle 
until $t'$, since even if $T$-robot $b$ observes $M$-robot $a$ between $t'$ and $(t')^+(a,M_E)$, $b$ does not change its color, $\hat{t}=(t')^+(a,M_E)$ becomes a cs-time
such that $[(S,S)]_{t+1} \stackrel{\mathrm{a}}{\mathrm{\rightarrow}} [(M,T)]_{\hat{t}}$.


{\bf Case (2)}: In this case, if we consider $t'$ and $t^+(a,C)$ as $t^+(a,C)$ and $t'$ in the case (1), respectively, it can be reduced to {\bf Case (1)}.

Thus, this lemma holds. $\qed$
\end{proof}

\begin{figure}[H]
\begin{center}
\includegraphics[width=70mm]{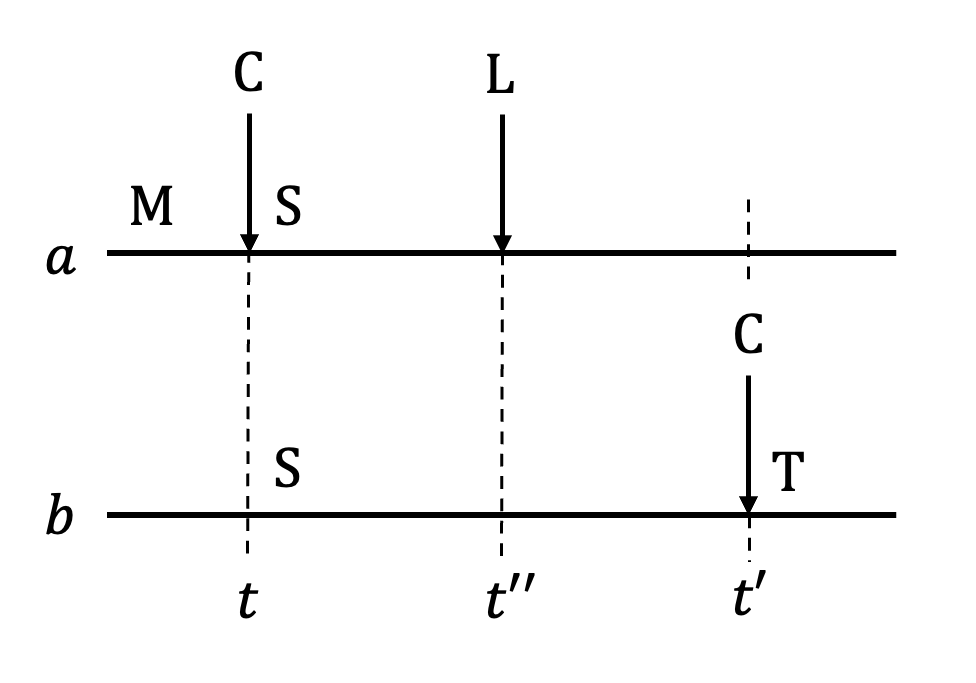}
\caption{The situation in Lemma~\ref{lem:ST->}.}\label{fig:ST->}
\end{center}
\end{figure}

\begin{figure}[H]
\begin{center}
\includegraphics[width=122mm]{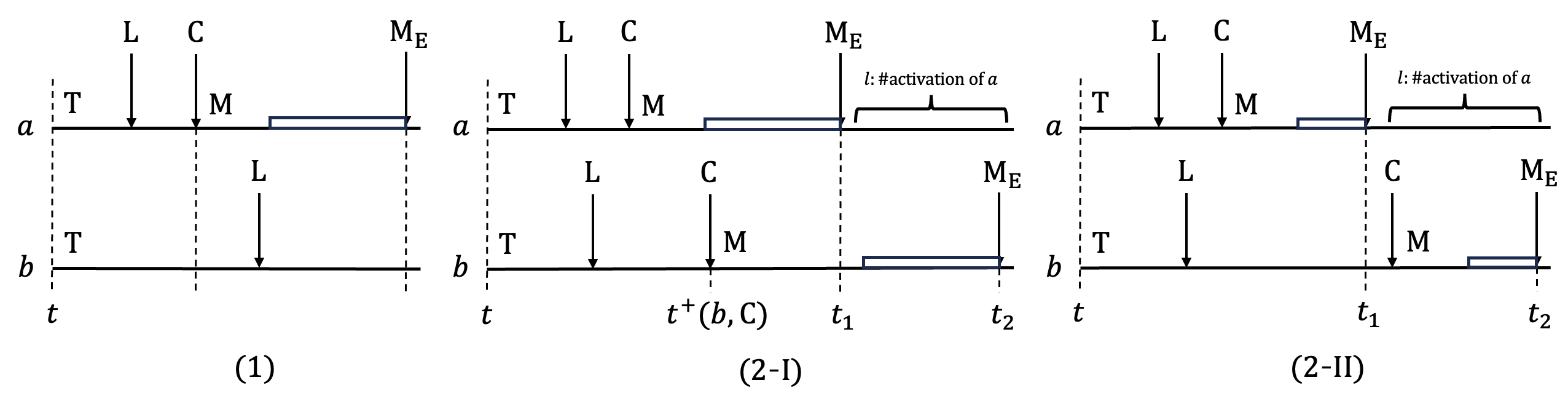}
\caption{The proof of Lemma~\ref{lem:TT}.}\label{fig:TT}
\end{center}
\end{figure}

\begin{figure}[H]
\begin{center}
\includegraphics[width=70mm]{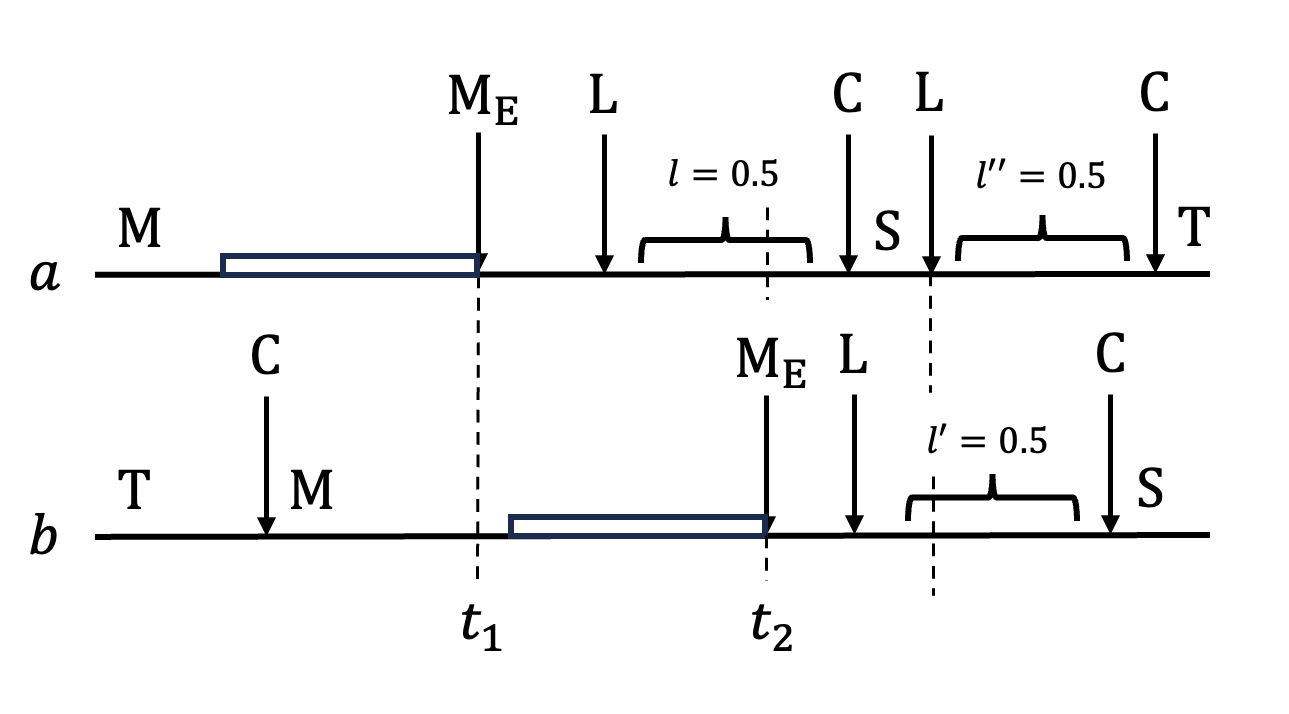}
\caption{$\ell=0.5$, $\ell'=0.5$, and $\ell''=0.5$ in (2-II).}\label{fig:ellell'ell''}
\end{center}
\end{figure}

\begin{lemma}\label{lem:TT}
    Let configuration at $cs$-time $t$ be $[(T,T)]_t$. Then there exists a $cs$-time $t'>t$ such that exactly one of the followings holds:
    \begin{enumerate}
    \item[(1)] $[(T,T)]_t \stackrel{\mathrm{(a,b)}}{\mathrm{\rightarrow}} [(M,M)]_{t'}$,  
    $[(T,T)]_t \stackrel{\mathrm{(a,b)}}{\mathrm{\rightarrow}} [(S,S)]_{t'}$, or
    $[(T,T)]_t \stackrel{\mathrm{(a,b)}}{\mathrm{\rightarrow}} [(T,T)]_{t'}$, 
    \item[(2)] $[(T,T)]_t \stackrel{\mathrm{(a,b)}}{\mathrm{\rightarrow}} [(S,M)]_{t'}$, or $[(T,T)]_t \stackrel{\mathrm{(a,b)}}{\mathrm{\rightarrow}} [(M,S)]_{t'}$,
    \item[(3)] $[(T,T)]_t \stackrel{\mathrm{(a,b)}}{\mathrm{\rightarrow}} [(T,M)]_{t'}$, or $[(T,T)]_t \stackrel{\mathrm{(a,b)}}{\mathrm{\rightarrow}} [(M,T)]_{t'}$,
    \item[(4)] $[(T,T)]_t \stackrel{\mathrm{(a,b)}}{\mathrm{\rightarrow}} [(T,S)]_{t'}$, or $[(T,T)]_t \stackrel{\mathrm{(a,b)}}{\mathrm{\rightarrow}} [(S,T)]_{t'}$,
    \item[(5)] $[(T,T)]_t \stackrel{\mathrm{(a,b),a}}{\mathrm{\rightarrow}} [(M,T)]_{t'}$, or $[(T,T)]_t \stackrel{\mathrm{(a,b),b}}{\mathrm{\rightarrow}} [(T,M)]_{t'}$,
    \item[(6)] $[(T,T)]_t, \stackrel{\mathrm{a}}{\mathrm{\rightarrow}} [(M,T)]_{t'}$, or $[(T,T)]_t  \stackrel{\mathrm{b}}{\mathrm{\rightarrow}} [(T,M)]_{t'}$. 
    \end{enumerate}
\end{lemma}

\begin{proof}
We can assume that $a$ is activated earlier than $b$\footnote{Otherwise, it is the symmetrical case and it can be treated similarly.}.   
\begin{enumerate}
    \item[Case 1:] ($t^+(a,C) \leq t^+(b,L)$ (see Fig.~\ref{fig:TT}(1)))
    
    Since $T$-robot $a$ observes $T$-robot $b$, $a$ changes its color to $M$ and performs $\cal{P}$. And since $b$ observes $M$-robot $a$, $b$ does not change its color (Lemma~\ref{lem:TT}($(\alpha,\beta)=(T,M)$),
    setting $t'=t^+(a,M_E)$ $t'$ becomes a $cs$-time and it holds that $[(T,T)]_t, \stackrel{\mathrm{a}}{\mathrm{\rightarrow}} [(M,T)]_{t'}$. ($(6)$)\footnote{Note that this result holds when $t^+(b,L) > t'$.}. 
    Symmetrically, it holds that $[(T,T)]_t  \stackrel{\mathrm{b}}{\mathrm{\rightarrow}} [(T,M)]_{t'}$.
    \item[Case 2:] (($t^+(a,L) \leq t^+(b,L) \leq t^+(a,C)$ (see Fig.~\ref{fig:TT}(2-I),(2-II)))

     Both robots change their color to $M$ and perform $\cal{P}$ with the same snapshot. Let $t_1=\min(t^+(a,M_E),t^+(b,M_E))$ and $t_2=\max(t^+(a,M_E),t^+(b,M_E))$. Wlog, we assume that $t^+(a,M_E)< t^+(b,M_E)$, and  then $t_1=t^+(a,M_E)$ and $t_2=t^+(b,M_E)$.
     There are two subcases (2-I) $t^+(b,C) < t^+(a,M_E)$ and (2-II) $t^+(b,C) \geq t^+(a,M_E)$.

     {\bf (2-I)} If $a$ is activated between $[t_1..t_2]$, then $M$-robot $a$ observes $M$-robot $b$. 
     Let $\ell$ be the number of activation of the robot $a$ completing the cycle between $[t_1..t_2]$.
     \begin{itemize}
     \item If $\ell=0$, letting $t'=t_2$, $[(T,T)]_t \stackrel{\mathrm{(a,b)}}{\mathrm{\rightarrow}} [(M,M)]_{t'}$($(1)$). 

     If $\ell=0$ but $a$ performs only $Look$-operation between $[t_1..t_2]$ (it is denoted as $\ell=0.5$), that is, $(t_1)^+(a,C)>t_2$, let $t_3=(t_1)^+(a,C)$ and $\ell'$ be the number of activation of the robot $b$ completing the cycle between $[t_2..t_3]$. 
     If $\ell'=0$ and $\ell' \geq 1$, setting $t'=t_3$, we can show that $[(T,T)]_t \stackrel{\mathrm{(a,b)}}{\mathrm{\rightarrow}} [(S,M)]_{t'}$($(2)$). Symmetrically, it holds that $[(T,T)]_t  \stackrel{\mathrm{b}}{\mathrm{\rightarrow}} [(M,S)]_{t'}$.  $[(T,T)]_t \stackrel{\mathrm{(a,b)}}{\mathrm{\rightarrow}} [(S,T)]_{t'}$($(4)$), respectively. In the case that $\ell'=0.5$,
     let $\ell''$ be the number of activation of the robot $a$ completing the cycle between $[t_3..t_4]$, where $t_4=(t_3)^+(b,C)$. If $\ell''=0$ and $\ell'' \geq 1$, setting $t'=t_4$, it holds that $[(T,T)]_t \stackrel{\mathrm{(a,b)}}{\mathrm{\rightarrow}} [(S,S)]_{t'}$($(1)$), $[(T,T)]_t \stackrel{\mathrm{(a,b)}}{\mathrm{\rightarrow}} [(T,S)]_{t'}$($(4)$), respectively. 
     The case that $\ell''=0.5$  but $a$ performs only $Look$-operation between $[t_3..t_4]$
     is left (Fig.~\ref{fig:ellell'ell''}) and so let $t_5=(t_4)^+(a,C)$.
     If $a$ is not activated between $[t_4..t_5]$, letting $t'=t_5$, it holds that  $[(T,T)]_t \stackrel{\mathrm{(a,b)}}{\mathrm{\rightarrow}} [(T,S)]_{t'}$($(4)$). Otherwise,
     using Lemma~\ref{lem:ST->} we can obtain that $[(T,T)]_{t} \stackrel{(a,b)}{\mathrm{\rightarrow}} [(T,T)]_{t'}$($(1)$), $[(T,T)]_{t} \stackrel{\mathrm{(a,b) a}}{\mathrm{\rightarrow}}[(M,T)]_{t'}$($(5)$), or $[(T,T)]_{t} \stackrel{\mathrm{(a,b) b}}{\mathrm{\rightarrow}}[(T,M)]_{t'}$($(5)$), where $t'=\hat{t}$ in Lemma~\ref{lem:ST->}. 
     
     \item If $\ell=1$ and one cycle completes between $[t_1..t_2]$,  since $M$-robot $a$ observes $M$-robot $b$, $a$ changes its color to $S$, letting $t'=t_2$, $[(T,T)]_t \stackrel{\mathrm{(a,b)}}{\mathrm{\rightarrow}} [(S,M)]_{t'}$.($(2)$).            
      If $\ell=1$ but $a$ performs only the second $Look$-operation between $[t_1..t_2]$,  $M$-robot $a$ changes its color to $S$ at the first $Comp$-operation and changes its color to $T$ at the second $Comp$-operation at $t'>t_2$. If $M$-robot $b$ is not activated between $[t_2..t']$, then   $[(T,T)]_t \stackrel{\mathrm{(a,b)}}{\mathrm{\rightarrow}} [(T,M)]_{t'}$.($(3)$). Symmetrically, it holds that $[(T,T)]_t  \stackrel{\mathrm{b}}{\mathrm{\rightarrow}} [(M,T)]_{t'}$.  Otherwise, even if $M$-robot $b$ observes $S$-robot $a$, $b$ does not change its color. Thus, it is the same as the first case.                  
                         
     \item If $\ell \geq 2$, $M$-robot $a$ changes its color to $S$ at the first $Comp$-operation and changes its color to $T$ at the second $Comp$-operation. Since $T$-robot $a$ does not change its color when seeing $M$-robot $b$, $t_2$ becomes a $cs$-time and it holds that $[(T,T)]_t \stackrel{\mathrm{(a,b)}}{\mathrm{\rightarrow}} [(T,M)]_{t_2}$($(3)$).
     \end{itemize}
     
     {\bf (2-II)} The case  that $t^+(b,C)<(t_1)^+(a,L)$ is the same as the case  {\bf (2-I)}. In the case that $t^+(b,C) \geq (t_1)^+(a,L)$, 
     If $M$-robot $a$ is not activated between $[t_1..t_2]$, setting $t'=t_2$, it hold that $[(T,T)]_t \stackrel{\mathrm{(a,b)}}{\mathrm{\rightarrow}} [(M,M)]_{t'}$($(1)$).  Otherwise,
     $M$-robot $a$ observes $T$-robot or $M$-robot $b$ and changes its color to $S$ at $(t_1)^+(a,C)=t_C$.
     There are two cases (a) $t_C \leq t_2$ and (b) $t_C > t_2$:
     \begin{itemize}
     \item (a) $t_C\leq t_2$:
     Let $\ell$ be the number of activation of the robot $a$ completing the cycle between $(t_C..t_2]$. If $\ell=0$, $t_2$ is a $cs$-time and $[(T,T)]_t \stackrel{\mathrm{(a,b)}}{\mathrm{\rightarrow}} [(S,M)]_{t_2}$($(2)$). 
     Otherwise ($\ell \geq 1$), since $S$-robot $a$ becomes a $T$-robot, $t_2$ is a $cs$-time and $[(T,T)]_t \stackrel{\mathrm{(a,b)}}{\mathrm{\rightarrow}} [(T,M)]_{t_2}$($(3)$). 
     \item (b) $t_C >t_2$:
     This case is the same as that of $\ell=0.5$ and $\ell'=0.5$ in (2-I).
     
    \end{itemize}                   
\end{enumerate}
$\qed$
\end{proof}

\begin{lemma}\label{lem:MM}
    Let configuration at $cs$-time $t$ be $[(M,M)]_t$. Then there exists a $cs$-time $t'>t$ such that exactly one of the followings holds:
    \begin{enumerate}
    \item[(1)] $[(M,M)]_t \stackrel{}{\mathrm{\rightarrow}} [(S,S)]_{t'}$, or
    $[(M,M)]_t \stackrel{}{\mathrm{\rightarrow}} [(T,T)]_{t'}$,
    \item[(2)] $[(M,M)]_t \stackrel{}{\mathrm{\rightarrow}} [(S,T)]_{t'}$, or $[(M,M)]_t \stackrel{}{\mathrm{\rightarrow}} [(T,S)]_{t'}$,
    \item[(3)] $[(M,M)]_t \stackrel{}{\mathrm{\rightarrow}} [(S,M)]_{t'}$, or $[(M,M)]_t, \stackrel{}{\mathrm{\rightarrow}} [(M,S)]_{t'}$,
    \item[(4)] $[(M,M)]_t \stackrel{}{\mathrm{\rightarrow}} [(T,M)]_{t'}$, or $[(M,M)]_t, \stackrel{}{\mathrm{\rightarrow}} [(M,T)]_{t'}$,
    \item[(5)] $[(M,M)]_t \stackrel{a}{\mathrm{\rightarrow}} [(M,T)]_{t'}$, or $[(M,M)]_t, \stackrel{b}{\mathrm{\rightarrow}} [(T,M)]_{t'}$. 
    \end{enumerate}
\end{lemma}
\begin{proof}
   There are two cases (1) $t^+(a,L) \leq t^+(b,L) \leq t^+(a,C)$, and (2) $t^+(a,C) < t^+(b,L)$.
\begin{itemize}
    \item[(1)] If $t^+(b,C) < t^+(a,C)$, let $\ell$ be the number of activation of $b$ completing the cycle between $[t^+(b,C)..t^+(a,C)]$. If $\ell=0$ and $\ell \geq 1$, setting $t'=t^+(a,C)$, we can show that $[(M,M)]_t \stackrel{\mathrm{}}{\mathrm{\rightarrow}} [(S,S)]_{t'}$($(1)$), $[(M,M)]_t \stackrel{\mathrm{}}{\mathrm{\rightarrow}} [(T,S)]_{t'}$($(2)$), respectively. Otherwise ($\ell=0.5$), using Lemma~\ref{lem:ST->} we can obtain that $[(M,M)]_{t} \stackrel{}{\mathrm{\rightarrow}} [(T,T)]_{t'}$($(1)$), $[(T,T)]_{t} \stackrel{\mathrm{a}}{\mathrm{\rightarrow}}[(M,T)]_{t'}$($(5)$), or $[(T,T)]_{t} \stackrel{\mathrm{b}}{\mathrm{\rightarrow}}[(T,M)]_{t'}$($(5)$), where $t'=\hat{t}$ in Lemma~\ref{lem:ST->}. 
    
    If $t^+(b,C) \geq t^+(a,C)$, exchanging $a$ and $b$, it is the same situation as the former case.
    \item[(2)] Let $\ell$ be the number of activation of $a$ completing the cycle between \\$[t^+(a,C)..t^+(b,C)=t_C]$. If $\ell=0$ and $\ell \geq 1$, setting $t'=t^+(b,C)$, we can show that $[(M,M)]_t \stackrel{\mathrm{}}{\mathrm{\rightarrow}} [(S,M)]_{t'}$($(3)$), $[(M,M)]_t \stackrel{\mathrm{}}{\mathrm{\rightarrow}} [(T,M)]_{t'}$($(4)$), respectively. Otherwise ($\ell=0.5$), since $M$-robot $b$ does not change its color when observing $S$-robot $a$, setting $t'=(t_C)^+(a,C)$, $[(M,M)]_t \stackrel{\mathrm{}}{\mathrm{\rightarrow}} [(T,M)]_{t'}$($(4)$).
\end{itemize}
 $\qed$
\end{proof}

\begin{figure}[H]
\begin{center}
\includegraphics[width=70mm]{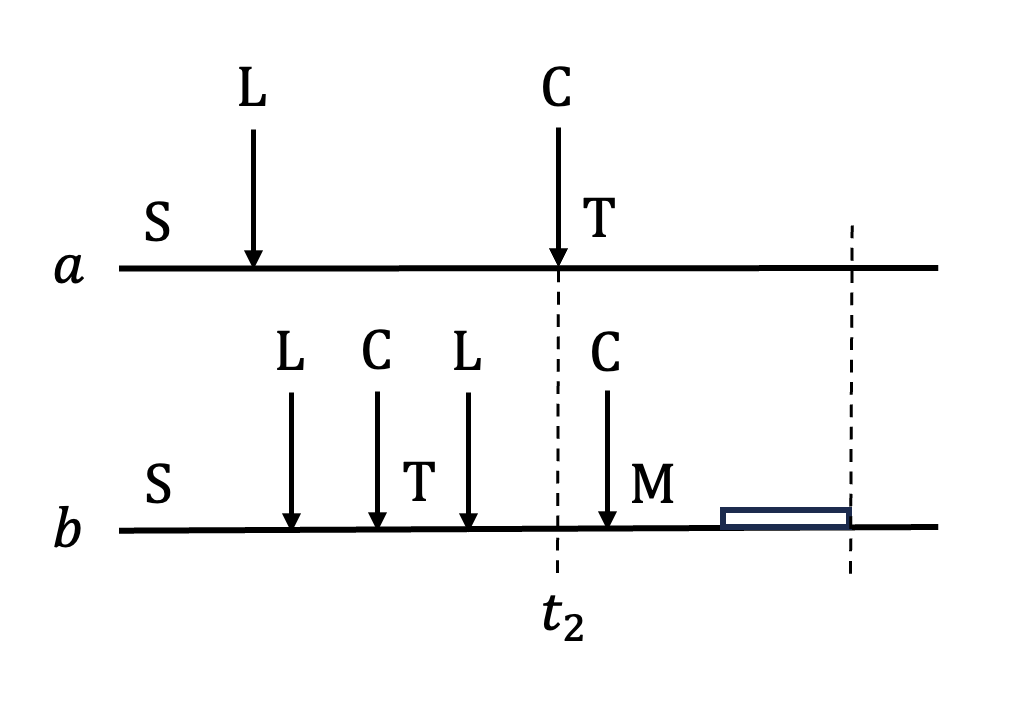}
\caption{The proof of Lemma~\ref{lem:SS}.}\label{fig:ell=1.5}
\end{center}
\end{figure}
\begin{lemma}\label{lem:SS}
    Let configuration at $cs$-time $t$ be $[(S,S)]_t$. Then there exists a $cs$-time $t'>t$ such that exactly one of the followings holds:
    \begin{enumerate}
    \item[(1)] $[(S,S)]_t \stackrel{}{\mathrm{\rightarrow}} [(T,T)]_{t'}$, 
    \item[(2)] $[(S,S)]_t \stackrel{\mathrm{b}}{\mathrm{\rightarrow}} [(T,M)]_{t'}$, or $[(S,S)]_t \stackrel{\mathrm{a}}{\mathrm{\rightarrow}} [(M,T)]_{t'}$,
    \item[(3)] $[(S,S)]_t \stackrel{}{\mathrm{\rightarrow}} [(T,S)]_{t'}$, or $[(S,S)]_t, \stackrel{}{\mathrm{\rightarrow}} [(S,T)]_{t'}$. 
    \end{enumerate}
\end{lemma}
\begin{proof}
 If $t^+(b,L) \geq t^+(a,C)$, it holds that  $[(S,S)]_t \stackrel{}{\mathrm{\rightarrow}} [(T,S)]_{t'}$($(3)$). Otherwise ($ t~+(a<L) <t^+(b,L)=t_1 < t~+(a,C)=t_2$),  let $\ell$ be the number of activation of the robot $b$ completing the cycle between $[t_1..t_2]$. If $\ell=1$, it holds that $[(S,S)]_t \stackrel{}{\mathrm{\rightarrow}} [(T,T)]_{t'}$($(1)$). If $\ell \geq 2$, $b$ becomes an $M$-robot and performs $\cal{P}$ by $t_2$. Since $M$-robot $b$ does not change its color when seeing $S$-robot $a$, setting $t'=t_2$, $[(S,S)]_t \stackrel{\mathrm{b}}{\mathrm{\rightarrow}} [(T,M)]_{t'}$($(2)$).
 If $\ell=1.5$ (Fig.~\ref{fig:ell=1.5}), we have two cases, (a) $T$-robot $a$ observes the $M$-robot $b$ after $t_2$ and (b) $T$-robot $a$  observes the $T$-robot $b$ after $t_2$. 
 \begin{itemize}
     \item[(a)] Since $T$-robot $a$ does not change its color when seeing $M$-robot $b$, setting $t'=(t_2)^+(b,M_E)$, $[(S,S)]_t \stackrel{\mathrm{b}}{\mathrm{\rightarrow}} [(T,M)]_{t'}$($(2)$).
     \item[(b)] In this case $t_2$ can be considered as a $cs$-time. Thus setting $t'=t_2$, it holds that $[(S,S)]_t \stackrel{}{\mathrm{\rightarrow}} [(T,T)]_{t'}$($(1)$).
 \end{itemize}   $\qed$
\end{proof}
\begin{lemma}\label{lem:TM-ST-SM}

  \begin{enumerate}
  
  \item[(1)] $[(T,M)]_t \stackrel{}{\mathrm{\rightarrow}} [(T,S)]_{t'}$, and  $[(M,T)]_t \stackrel{}{\mathrm{\rightarrow}} [(S,T)]_{t'}$.
  \item[(2)] $[(S,M)]_t \stackrel{}{\mathrm{\rightarrow}} [(T,M)]_{t'}$, and $[(M,S)]_t \stackrel{}{\mathrm{\rightarrow}} [(M,T)]_{t'}$. 
  \item[(3)] $[(S,T)]_t \stackrel{\mathrm{b}}{\mathrm{\rightarrow}} [(S,M)]_{t'}$, or $[(S,T)]_t \stackrel{\mathrm{b}}{\mathrm{\rightarrow}} [(T,M)]_{t'}$.
  \item[(4)] $[(T,S)]_t \stackrel{\mathrm{a}}{\mathrm{\rightarrow}} [(M,S)]_{t'}$, or $[(T,S)]_t \stackrel{\mathrm{a}}{\mathrm{\rightarrow}} [(M,T)]_{t'}$.
  \end{enumerate}
\end{lemma}
\begin{proof}
\begin{itemize}
    \item For (1) and (2) they are immediately obtained by Lemma~\ref{lem:nochangecolor}.
    \item For (3) and (4), they are obtained by using the similar method to the case (1) of Lemma~\ref{lem:MM}.
\end{itemize}   \qed 
\end{proof}
By using Lemmas~\ref{lem:TT}-\ref{lem:TM-ST-SM}, We can show that Figure~\ref{fig:SIM2}-(b) is the correct transition diagram for SIM-2$^{RS}_A$.

\setcounterref{theoremduplicate}{th:Sim-RS-by-AS-2robots}
\addtocounter{theoremduplicate}{-1}

\begin{theoremduplicate}[Reprint of Theorem~\ref{th:Sim-RS-by-AS-2robots} on 
page~\pageref{th:Sim-RS-by-AS-2robots}]\label{duplicate:th:Sim-RS-by-AS-2robots}

$\forall R \in \mathcal R_2,  {\OB^{RS}}(R) \subseteq \LU_3^{A}(R).$
\end{theoremduplicate}

\begin{proof}
The initial configuration at time $0$ is $[(T,T)]_0$. 

{\bf Case 1}:
The transitions of Lemma~\ref{lem:TT}(1), Lemma~\ref{lem:MM}(1), or Lemma~\ref{lem:SS}(1) occur continuously, the loop of $(T,T) \stackrel{}{\mathrm{\rightarrow}}  (M,M) \rightarrow (S,S) \rightarrow (T,T)$ in Fig.~\ref{fig:SIM2} repeats, and it ends at $(T,T)$, $(M,M)$, or $(S,S)$. Since $a$ and $b$ perform $\cal{P}$ simultaneously once in every loop, SIM$2^{RS}_A$ simulates "\FSY" phase correctly.

{\bf Case 2}:
In the loop of {\bf Case 1}, some transition except Lemma~\ref{lem:TT}-\ref{lem:SS} (1) occurs from $(T,T)$.

{\bf Case 2-1}: When the transition of Lemma~\ref{lem:TT}(2) occurs at some $cs$-time $t_i$, the transition $[(T,T)]_{t_i} \stackrel{\mathrm{(a,b)}}{\mathrm{\rightarrow}} [(S,M)]_{t_i'}$ occurs and $a$ and $b$ perform $\cal{P}$ simultaneously once between $[t_i..t_i']$\footnote{Since the other case symmetrical, we can prove similarly.}. Since the configuration at $t_i'$ is $[(S,M)]_{t_i'}$, the transition of configurations after 
$t_i'$ is 

$[(S,M)]_{t_i'} 
\stackrel{}{\mathrm{\rightarrow}} [(T,M)]_{t_1}$(Lemma~\ref{lem:TM-ST-SM}(2))

$[(T,M)]_{t_1} 
\stackrel{}{\mathrm{\rightarrow}} [(T,S)]_{t_2}$(Lemma~\ref{lem:TM-ST-SM}(1))

$[(T,S)]_{t_2} 
\stackrel{\mathrm{a}}{\mathrm{\rightarrow}} [(M,S)]_{t_3}$(Lemma~\ref{lem:TM-ST-SM}(4))

$[(M,S)]_{t_3} 
\stackrel{}{\mathrm{\rightarrow}} [(M,T)]_{t_4}$(Lemma~\ref{lem:TM-ST-SM}(2))

$[(M,T)]_{t_4} 
\stackrel{}{\mathrm{\rightarrow}} [(S,T)]_{t_5}$(Lemma~\ref{lem:TM-ST-SM}(1))

$[(S,T)]_{t_5} 
\stackrel{\mathrm{b}}{\mathrm{\rightarrow}} [(S,M)]_{t_6}$(Lemma~\ref{lem:TM-ST-SM}(3)).

Since this loop contains $a$ performing $\cal{P}$ followed by $b$ performing $\cal{P}$, SIM$2^{RS}_A$ simulates "\RSY" phase ($a$ and $b$ performs alternately) correctly.

{\bf Case 2-2}:
When the transition of Lemma~\ref{lem:TT}(3) occurs at some $cs$-time $t_i$, the transition $[(T,T)]_{t_i} \stackrel{\mathrm{b}}{\mathrm{\rightarrow}} [(T,M)]_{t_i'}$ occurs and $b$ performs $\cal{P}$ once between $[t_i..t_i']$. Since we can consider the loop of {\bf Case 2-1} starting at $[(T,M)]_{t_i'}$, the order of performing $\cal{P}$ is $b$, $a$, $b$, $\cdots$. Thus SIM-2$^{RS}_A$ simulates "\RSY" phase correctly.

The other cases, the transitions of Lemma~\ref{lem:TT} (3)$~$(6) from $(T,T)$, Lemma~\ref{lem:MM} (2)$~$(5), and Lemma~\ref{lem:SS} can be shown similarly. \qed
\end{proof}

Also it is verified that SIM-$2^{RS}_A$ works correctly from any initial configuration shown in~\ref{fig:SIM2}-(b), that is, SIM-$2^{RS}_A$ is self-stabilizing.


\begin{theorem}
Protocol SIM-$2^{RS}_A$ is correct and self-stabilizing.
\end{theorem}

Thus, we have Theorem~\ref{th:Sim-RS-by-AS-2robots} and the following corollary.
\begin{corollary}
\label{co:Sim-RS-by-AS-2robots}
$\forall R \in \mathcal R_2,  {\LU_k^{RS}}(R) \subseteq \LU_{3k}^{A}(R).$
\end{corollary}


\subsection{The impossibility of the simulation with $2$ colors}

In this subsection, we show that any simulation of two $\OB^{RS}$ robots by two $\LU^{S}$ robots with two colors is impossible. Thus, the three-color simulation in the preceding subsection is optimal with respect to the number of colors.  

\begin{algorithm}[H]
    
    \caption{Protocol SIM-2$^{RS}_S$ with two colors $X$ and $Y$}\label{alg4}
   
    State \textit{Look} 
    \begin{algorithmic}[0] 
    
    \STATE $ my.light $
    \STATE $ other.light:$ the other robot's light.
    \vspace{3mm}
    \end{algorithmic}
    
    State \textit{Compute} 
    \begin{algorithmic}[1]
    \STATE $my.des \leftarrow my.pos$
    \IF {$my.light = X$ {\bf and} $other.light = X$}
        \STATE action $\cal{A}$$_0$ 
        \STATE $my.light \leftarrow C_0$
    \ELSIF {$my.light = X$ {\bf and} $other.light = Y$}
        \STATE action $\cal{A}$$_1$ 
        \STATE $my.light \leftarrow C_1$
       
    \ELSIF {$my.light = Y$ {\bf and} $other.light = X$}
        \STATE action $\cal{A}$$_2$ 
        \STATE $my.light \leftarrow C_2$

    \ELSIF {$my.light = Y$ {\bf and} $other.light = Y$}
        \STATE action $\cal{A}$$_3$ 
        \STATE $my.light \leftarrow C_3$
    
    \ENDIF
    \vspace{3mm}
    \end{algorithmic}
 \end{algorithm}

Let $a$ and $b$ be two $\LU^S$ robots with two colors $X$ and $Y$. Since any action in any simulation protocol depends on only own color and the other's color, when $\alpha$-robot observes $\beta$-robot,
determining action $\cal{A}  \in \{$  "Execute $\cal{P}"$, "no action"$\}\in \}$ and the next color $\gamma$ defines a simulating protocol, where $\alpha, \beta, \gamma \in \{X,Y\}$ (Algorithm~\ref{alg4}).
For example, in Algorithm~\ref{alg4} setting $\cal{A}$$_0=$ execute $\cal{P}$, $\cal{A}$$_i=$ no action ($i=1,2,3$) and $C_0=Y$, $C_i=X (i=1,2,3)$,
one simulating protocol is defined. However, this protocol cannot simulate \RSY\ by $\LU^{S}$ robots, because considering a schedule that only $a$ robot is activated, $a$ performs $\cal{P}$ consecutively and so it violates \RSY.

There are $2^8$ simulating protocols including meaningless and these are all protocols which simulate two $\LU^{RS}$ robots by two $\LU^S$ robots with two colors, and we can verify that schedules all the simulating protocols produce violate \RSY.
Thus, we obtain the following theorem.

\begin{theorem}\label{th:optimalSIM-2RSA}
SIM-2$^{RS}_A$ is an optimal simulating protocol with respect to the number of colors.
\end{theorem}

%% file: conclusion.tex
In this paper, we discuss efficient protocols for simulating \RSY\ under \ASY\ or \SSY\ for $\LU$ robots. 
In particular, for the simulation of \RSY\ under \SSY\, we have significantly reduced the number of colors previously required. Also, in the simulation of \RSY\ under \ASY, we have achieved the simulation with the same number of colors as used in previous simulations of \SSY\ under \ASY. Furthermore, for the case of $n=2$, we have realized the simulation of \RSY\ under \ASY\ with an optimal number of colors. We have also shown that all our proposed protocols are self-stabilizing and their self-stabilization can be done without increasing the number of colors. An outstanding issue is the reduction of the number of colors needed for simulating $\LU$ robots under \RSY\ for $\FC$ robots.